\pdfoutput=1
\documentclass[11pt]{article} 
\def\arXiv{1}

\input{macros.sty}

\notarxiv{
	\usepackage[subtle, all=normal, mathspacing=tight, floats=tight, 
	sections=tight]{savetrees}
}

\title{Principal Component Projection and Regression in\\ Nearly Linear Time through Asymmetric SVRG
}

\arxiv{
	
\author{
Yujia Jin \\
Stanford University \\
\texttt{\href{mailto:yujiajin@stanford.edu}{yujiajin@stanford.edu}}
\and
	Aaron Sidford \thanks{Research supported in part by NSF CAREER Award CCF-1844855.}\\
Stanford University \\
\texttt{\href{mailto:sidford@stanford.edu}{sidford@stanford.edu}}
}

\date{}
}

\begin{document}

\maketitle

\begin{abstract}

Given a data matrix $\A \in \R^{n \times d}$, \emph{principal component projection (PCP)} and \emph{principal component regression (PCR)}, i.e. projection and regression restricted to the top-eigenspace of $\A$, are fundamental problems in machine learning, optimization, and numerical analysis. In this paper we provide the first algorithms that solve these problems in nearly linear time 
for fixed eigenvalue distribution and large  $n$. This improves upon previous methods which have superlinear running times when both the number of top eigenvalues and inverse gap between eigenspaces is large. %
We achieve our results by applying rational approximations to reduce PCP and PCR to solving asymmetric linear systems which we solve by a variant of SVRG. We corroborate these findings with preliminary empirical experiments.

\end{abstract}

\newpage

\section{Introduction}

PCA is one of the most fundamental algorithmic tools for analyzing large data sets. Given a data matrix $\A \in \R^{n \times d}$ and a parameter $k$ the classic \emph{principal component analysis} (PCA) problem asks to compute the top $k$ eigenvectors of $\A^\trans \A$. This is a core computational task in machine learning and often used for feature selection~\cite{malhi2004pca,song2010feature,pascoal2012robust}, data visualization~\cite{niedoba2014multi,metsalu2015clustvis}, and model compression~\cite{yin2015efficient}.

However, as $k$ becomes large, the running time of PCA can degrade. Even just writing down the output takes $\Omega(kd)$ time and the performance of many methods degrade with $k$.
This high-computational cost for exploring large-cardinality top-eigenspaces has motivated researchers to consider prominent tasks solved by PCA, for example \emph{principal component projection (PCP)} which asks to project a vector onto the top-$k$ eigenspace, and \emph{principal component regression (PCR)} which asks to solve regression restricted to this top-$k$ eigenspace (see \cref{sec:prob}).  

Recent work \cite{FMMS16,AL17} showed that the dependence on $k$ in solving PCP and PCR  can be overcome by instead depending on \emph{eigengap $\gamma$}, defined as the ratio between the smallest eigenvalue in the space projected onto and the largest eigenvalue of the space projected orthogonal to. These works replace the typical $\poly(k) \nnz(\A)$ dependence in runtime with a $\poly(1/\gamma) \nnz(\A)$ (at the cost of lower order terms), by reducing these problems to solving $\poly(1/\gamma)$ ridge-regression problems. %
Unfortunately, for large-scale %
 problems, as data-set sizes grow so too can $k$ and $1/\gamma$, yielding large super-linear running times for all previously known methods (see \cref{sec:comparison}). Consequently, this leaves the following fundamental open problem:\\
\\
\emph{Can we obtain  nearly linear running times for solving PCP and PCR to high precision, i.e. having running time $\tilde{O}(\nnz(\A))$ plus an additive term depending only on the eigenvalue distribution?}\\

The main contribution of the paper is an affirmative answer to this question. We design randomized algorithms that solve PCP and PCR with high probability in nearly linear time. Leveraging rational polynomial approximations we reduce these problems to solving asymmetric linear systems, which we solve by a technique we call \emph{asymmetric SVRG}. Further, we provide experiments demonstrating the efficacy of this method.

\subsection{Approach}
To obtain our results, critically we depart from the previous frameworks in~\citet{FMMS16,AL17} for solving PCP and PCR. These papers use polynomial approximations to the sign function to reduce PCP and PCR to solving ridge regression. Their runtime is limited by the necessary $\Omega(1/\gamma)$ degree for polynomial approximation of the sign function shown by \citet{eremenko2007uniform}. Consequently, to obtain nearly linear runtime, a new insight is required.

In our paper, we instead consider rational approximations to the sign function and show that these efficiently reduce PCP and PCR to solving a particular class of squared linear systems.
The closed form expression for the best rational approximation to sign function was given by Zolotarev~\cite{Z77} and has recently been proposed for matrix sign approximation~\cite{NF16}. The degree of such rational functions is logarithmic in $1/\gamma$, leading to much fewer linear systems to solve. While the \emph{squared systems} $[(\A^\top\A-c\I)^2+\mu^2\I]\xx=\bb,\mu>0$ induced by this rational approximation are computationally more expensive to solve, as compared with simple ridge regression problems $\left[\A^\top\A+\mu\I\right]\xx=\bb,\mu>0$, interestingly, we show that  these systems can still be solved in nearly linear time for sufficiently large matrices. As a by-product of this analysis, we also obtain an efficient algorithm for leveraging linear system solvers to apply the square-root of a positive semidefinite (PSD) %
matrix to a  vector, where we call a matrix $\M$ positive semidefinite and denote $\M\succeq\0$  when $\forall \xx,\xx^\top\M\xx\ge0$.

We believe the solver we develop for these squared systems is of independent interest. Our solver is a variant of the stochastic variance-reduced gradient descent algorithm (SVRG)~\cite{johnson2013accelerating} modified to solve asymmetric linear systems. Our iterative method can be viewed as an instance of the variance-reduced algorithm for monotone operators discussed in Section 6 of~\citet{pal16}, with a more careful analysis of the error. We also combine this method with approximate proximal point~\cite{FGKS15} or  Catalyst~\cite{LMH15} to obtain accelerated variants.

The conventional wisdom when solving asymmetric systems $\M \xx = \bb$ that are not positive semidefinite (PSD), i.e. $\M\nsucceq\0$, is to instead solve its PSD counterpart $\M^\top \M \xx = \M^\top \bb$. However, this operation can greatly impair the performance of stochastic methods, e.g. SVRG~\cite{johnson2013accelerating}, SAG~\cite{schmidt2017minimizing}, etc. 
(See \cref{sec:SVRG}.)
The solver developed in this paper constitutes one of few known cases where transforming it into asymmetric system solving enables better algorithm design and thus provides large savings (see \cref{cor:square_solver_main}.) Ultimately, we believe this work on SVRG-based methods outside of convex optimization as well as our improved PCP and PCR algorithms may find further impact.

\subsection{The Problems}
\label{sec:prob}

Here we formally define the PCP (\cref{dfn:PCP}), PCR (\cref{dfn:PCR}), Squared Ridge Regression (\cref{dfn:square}), and Square-root Computation (\cref{dfn:square-root}) problems we consider throughout this paper. Throughout, we let $\A\in\R^{n\times d}$  $(n\ge d)$ denote  a data matrix where each row $\aaa_i\in\R^n$ is viewed as a datapoint. Our algorithms typically manipulate the positive semidefinite (PSD) matrix $\A^\top\A$.  We denote the eigenvalues of $\A^\top\A$ as $\lambda_1\ge\lambda_2\ge\cdots\ge\lambda_d\ge0$ and corresponding eigenvectors as $\nnu_1,\nnu_2,\cdots,\nnu_d\in \R^d$, i.e.  $\A^\top\A=\V\LLambda\V^\top$ with $\V\defeq(\nnu_1,\cdots,\nnu_d)^\top$ and $\LLambda\defeq\diag(\lambda_1,\cdots,\lambda_d)$.

Given eigenvalue threshold $\lambda\in(0,\lambda_1)$ we define $\PP_\lambda\defeq(\nnu_1,\cdots,\nnu_k)(\nnu_1,\cdots,\nnu_k)^\top$ as a projection matrix projecting any vector onto the top-$k$ eigenvectors of $\A^\top\A$, i.e.  $\mathrm{span}\{\nnu_1,\nnu_2,\cdots,\nnu_k\}$, where $\lambda_k$ is the minimum eigenvalue of  $\A^\top\A$ no smaller than $\lambda$, i.e. $\lambda_k\ge\lambda>\lambda_{k+1}$. Without specification $\|\cdot\|$ is the standard $\ell_2$-norm of vector or matrix.

Given $\gamma\in(0,1)$, the goal of a PCP algorithm is to project any given vector $\vv=\sum_{i\in[d]}\alpha_i\nnu_i$ in a desired way: mapping $\nnu_i$ of $\A^\top\A$ with eigenvalue $\lambda_i$ in $[\lambda(1+\gamma),\infty)$ to itself, eigenvector $\nnu_i$ with eigenvalue $\lambda_i$ in $[0,\lambda(1-\gamma)]$ to $\0$, and any eigenvector $\nnu_i$ with eigenvalue $\lambda_i$ in between the gap to  anywhere between $\0$ and $\nnu_i$. Formally, we define the PCP as follows.

\begin{definition}[Principal Component Projection]
\label{dfn:PCP}

 The principal component projection (PCP) of $\vv\in\R^d$ at threshold $\lambda$ is $\vv_\lambda^*=\PP_\lambda \vv$. Given threshold $\lambda$ and eigengap $\gamma$, an algorithm $\mathcal{A}_{\PCP}(\vv,\eps,\delta)$ is an $\epsilon$-approximate PCP algorithm if with probability $1-
 \delta$, its output satisfies following:
 \begin{align}
 & \bullet \|\PP_{(1+\gamma)\lambda}(\mathcal{A}_{\PCP}(\vv)-\vv)\|\le\eps\|\vv\|;\nonumber\\ 
& \bullet 
 	\|(\I-\PP_{(1-\gamma)\lambda})\mathcal{A}_{\PCP}(\vv)\|\le\eps\|\vv\|
\label{cond:PCP}\\
& \bullet 
\|(\PP_{(1+\gamma)} - \PP_{(1-\gamma)\lambda}) (\mathcal{A}_{\PCP}(\vv) - \vv)
\|\le
\|(\PP_{(1+\gamma)} - \PP_{(1-\gamma)\lambda}) \vv \|
+ \eps\|\vv\|
\nonumber
 \end{align}
	
\end{definition}

The goal of a PCR problem is to solve regression restricted to the particular eigenspace we are projecting onto in PCP. The resulting solution should have no correlation with eigenvectors $\nnu_i$ corresponding to $\lambda_i\le\lambda(1-\gamma)$, while being accurate for $\nnu_i$ corresponding to eigenvalues above $\lambda_i\ge\lambda(1+\gamma)$. Also, it shouldn't have too large correlation with $\nnu_i$ corresponding to eigenvalues between $(\lambda(1-\gamma),\lambda(1+\gamma))$. Formally, we define the PCR problem as follows.

\begin{definition}[Principal Component Regression]
\label{dfn:PCR}

The principal component regression (PCR) of an arbitrary vector $\bb\in\R^n$ at threshold $\lambda$ is $\xx_\lambda^*=\min_{\xx\in\R^d}\|\A\PP_\lambda \xx-\bb\|$. 
 Given threshold $\lambda$ and eigengap $\gamma$, an algorithm $\mathcal{A}_{\PCR}(\bb,\eps,\delta)$ is an $\epsilon$-approximate PCR algorithm if with probability $1-\delta$, its output satisfies following: %
 \begin{equation}
 \|(\I-\PP_{(1-\gamma)\lambda})\mathcal{A}_{\PCR}(\bb,\eps)\|\le\epsilon\|\bb\|
 \enspace \text{ and } \enspace
 \|\A\mathcal{A}_{\PCR}(\bb,\eps)-\bb\|\le\|\A \xx_{(1+\gamma)\lambda}^*-\bb\|+\eps\|\bb\|
 ~.
 \label{cond:PCR}
 \end{equation}
\end{definition}

We reduce PCP and PCR to solving squared linear systems. The solvers we develop for this squared regression problem defined below we believe are of independent interest.

\begin{definition}[Squared Ridge Regression Solver]
\label{dfn:square}
Given $c\in[0,\lambda_1]$~\footnote{We remark that when $c<0$ (or $c>\lambda_1$), we can preprocess the problem by solving $(\A^\top\A-c\I)+\mu I$ (or $(c\I-\A^\top\A)+\mu \I$) twice, which is known to have efficient solvers~\cite{FMMS16,GHJ+16} enjoying provably better runtime guarantees than what we've shown for the harder (non-PSD) case $c\in[0,\lambda_1]$. }, $\vv\in\R^d$, we consider a squared ridge regression problem where exact solution is $\xx^*=((\A^\top \A-c\I)^2+\mu^2\I)^{-1}\vv$. We call an algorithm $\RidgeSquare(\A,c,\mu^2,\vv,\eps,\delta)$ an $\eps$-approximate squared ridge regression solver if with probability $1-\delta$ it returns a solution $\tilde{\xx}$ satisfying $\|\tilde{\xx}-\xx^*\|\le\eps\|\vv\|.$
\end{definition}

Using a similar idea of rational polynomial approximation, we also examine the problem of $\M^{1/2}\vv$ for arbitrarily given PSD matrix $\M$ to solving PSD linear systems approximately. 

\begin{definition}[Square-root Computation]%

\label{dfn:square-root}
Given a PSD matrix $\M\in\R^{n\times n}$ such that $\mu \I\preceq\M\preceq \lambda\I$ and $\vv\in\R^n$, an algorithm $\SR(\M,\vv,\eps,\delta)$ is an $\eps$-approximate square-root solver if with probability $1-\delta$ it returns a solution $\xx$ satisfying $\|\xx-\M^{1/2}\vv\|\le\eps\|\M^{1/2}\vv\|$.
\end{definition}

\subsection{Our Results}
\label{sec:results}

Here we present the main results of our paper, all proved in \cref{App:main}. For data matrix $\A\in\R^{n\times d}$, our running times are presented in terms of the following quantities.
\begin{itemize}
\item Input sparsity: $\nnz(\A)\defeq\text{ number of nonzero entries in }\A$;
\item Frobenius norm: $\|\A\|_\mathrm{F}^2\defeq \Tr(\A^\top\A)$;
\item Stable rank: $\mathrm{sr}(\A)\defeq \|\A\|_\mathrm{F}^2 / \|\A\|_2^2 = \|\A\|_\mathrm{F}^2/\lambda_1$;
\item Condition number of top-eigenspace: $\kappa\defeq\lambda_1/\lambda$. 
\end{itemize}
When presenting running times we use $\tilde{O}$ to hide polylogarithmic factors in the input parameters $\lambda_1 ,\gamma,\vv,\bb$, error rates $\eps$, and success probability $\delta$. %

For $\A\in\R^{n\times d}$ ($n\ge d$), $\vv\in\R^d$, $\bb\in\R^n$, without loss of generality we assume $\lambda_1\in[1/2,1]$\footnote{This can be achieved by getting a constant approximating overestimate $\tilde{\lambda}_1$ of $\A^\top\A$'s top eigenvector $\lambda_1$ through power method in $\tilde{O}(\nnz(\A))$ time, 
and consider $\A\leftarrow\A/\sqrt{\tilde{\lambda}_1},\lambda\leftarrow\lambda/\tilde{\lambda}_1,\bb\leftarrow\bb/\sqrt{\tilde{\lambda}_1}$ instead.} Given threshold $\lambda\in(0,\lambda_1)$ and eigengap $\gamma\in(0,2/3]$, the main results of this paper are the following new running times for solving these problems. %

\begin{theorem}[Principal Component Projection]
\label{thm:pcp_main}
For any $\eps\in(0,1)$, there is an $\eps$-approximate PCP algorithm (see \cref{dfn:PCP}) $\ISPCP(\A,\vv,\lambda,\gamma,\eps,\delta)$ specified in \cref{alg:ISPCP} with runtime 
	 $$\tilde{O} \left(\nnz(\A)+\sqrt{\nnz(\A)\cdot d\cdot\mathrm{sr}(\A)}\kappa/\gamma \right).$$
\end{theorem}

\begin{theorem}[Principal Component Regression]
\label{thm:pcr_main}
For any $\eps\in(0,1)$, there is an $\eps$-approximate PCR algorithm (see \cref{dfn:PCR}) $\ISPCR(\A,\bb,\lambda,\gamma,\eps,\delta)$ specified in \cref{alg:ISPCR} with runtime 
	 $$\tilde{O}\left(\nnz(\A)+\sqrt{\nnz(\A)\cdot d\cdot\mathrm{sr}(\A)}\kappa/\gamma\right).$$
\end{theorem}

We achieve these results by introducing a technique we call \emph{asymmetric SVRG} to solve squared systems $[(\A^\top\A-c\I)^2+\mu^2\I]\xx=\vv$ with $c\in[0,\lambda_1]$.  The resulting algorithm is closely related to the SVRG algorithm for monotone operators in~\citet{pal16}, but involves a more fine-grained error analysis. This analysis coupled with approximate proximal point~\cite{FGKS15} or Catalyst~\cite{LMH15} yields the following result (see \cref{sec:SVRG} for more details).

\begin{theorem}[Squared Solver]
	\label{thm:square_solver_main}
	For any $\eps\in(0,1)$, there is an $\eps$-approximate squared ridge regression solver (see \cref{dfn:square}) using $\AsySVRG(\M,\hat{\vv},\zz_0,\eps\|\vv\|,\delta)$ that runs in time 
$$\tilde{O}\left(\nnz(\A)+\sqrt{\nnz(\A)d\cdot\mathrm{sr}(\A)}\lambda_1/\mu\right).$$
\end{theorem}

When the eigenvalues of $\A^\top\A-c\I$ are bounded away from $0$, such a solver can be utilized to solve non-PSD linear systems in form $(\A^\top\A-c\I)\xx=\vv$ through preconditioning and considering the corresponding problem $(\A^\top\A-c\I)^2\xx=(\A^\top\A-c\I)\vv$ (see \cref{cor:square_solver_main}).

\begin{corollary}
	\label{cor:square_solver_main}
	Given $c\in[0,\lambda_1]$, and a non-PSD system $(\A^\top\A-c\I)\xx=\vv$ and an initial point $\xx_0$, for arbitrary $c$ satisfying $(\A^\top\A-c\I)^2\succeq\mu^2\I,\mu>0$, there is an algorithm returns with probability $1-\delta$ a solution $\widetilde{\xx}$ such that $\|\widetilde{\xx}-(\A^\top\A-c\I)^{-1}\vv\|\le\epsilon\|\vv\|$, within runtime 
	$\tilde{O}\bigl(\nnz(\A)+\sqrt{\nnz(\A)d\cdot\mathrm{sr}(\A)}\lambda_1/\mu\bigr)$.
\end{corollary}

Another byproduct of the rational approximation used in the paper is a nearly-linear runtime for computing an $\eps$-approximate square-root of PSD matrix  $\M\succeq \0$ applied to an arbitrary vector.%

\begin{theorem}[Square-root Computation]
	\label{thm:square_root_solver_main}
For any $\eps\in(0,1)$, given $\mu\I\preceq\M\preceq\lambda\I$, there is an $\eps$-approximate square-root solver (see \cref{dfn:square-root}) $\SR(\M,\vv,\eps,\delta)$ that runs in time 
$$\tilde{O}(\nnz(\M)+\mathcal{T})$$ 
where $\mathcal{T}$ is the runtime for solving $(\M+\kappa\I)\xx=\vv$ for arbitrary $\ \vv\in\R^n$ and $\kappa\in[\tilde{\Omega}(\mu),\tilde{O}(\lambda)]$.
\end{theorem}

\subsection{Comparison to Previous Work}
\label{sec:comparison}

\textbf{PCA and Eigenvector Computation}: PCA and eigenvector computation are well-known to be solvable using the classic power method, Chebyshev iteration, and Lanczos methods (see e.g.~\citet{golub2012matrix}). These deterministic methods all have an $\Omega(k \cdot \nnz(A))$ running time for computing the top $k$-eigenvectors, and are suitable for the high-accuracy approximation since they only depend logarithmically on accuracy $\eps$. On the other hand, there have also been several stochastic works which improve the running time for obtaining the top eigenvector~\cite{oja1982simplified,oja1985stochastic,clarkson2013low} at the cost of having a polynomial dependence on the desired accuracy $\epsilon$. This is a stochastic method designed in~\citet{allen2016lazysvd}, which has a $\Omega(k\cdot nd)$ running time but depend only logarithmically on $\eps$.  This method was built of  \cite{GHJ+16}, which provided the first stochastic algorithm that computes the top eigenvector in nearly-linear runtime.  The shift-and-invert method discussed in \cite{GHJ+16} broadly inspired the squared system solving in this paper, which can be viewed as shift-and-inverting at threshold $\lambda$ for the squared system $(\A^\top\A-\lambda\I)^2$. Note these methods all have an intrinsic linear dependence in $k$ for explicitly expressing the top-$k$ eigenvectors.

We remark that one can also use fast-matrix multiplication (FMM) to compute $\A^\top \A$ in $O(n d^\omega)$ time and compute SVD of this matrix in an additional $O\left(d^{\omega}\right)$ time, where $\omega < 2.379$~\cite{Williams12}  is the matrix multiplication constant, to speed up the time for PCA, PCR, and PCP. Given the well known practicality concerns of methods using fast matrix multiplication we focus much of our comparison on methods that do not use FMM.

\textbf{PCP and PCR}: The starting point for our paper is the work of \citet{FMMS16}, which provided the first nearly linear time algorithm for these problems when the eigengap is constant, by reducing the problem to finding the best polynomial approximation to sign function and solving a sequence of regression problems. This dependence of these methods on the eigengap was improved by \citet{AL17} through Chebyshev polynomials and then in \citet{musco2018stability} by Lanczos method. %
 These algorithms were first to achieve input sparsity for eigenspaces of any non-trivial size, but with super-linear running times whenever the eigenvalue-gap is super-constant. Instead of applying polynomial approximation, we use rational function approximations and reduce to different subproblems to get new algorithms with better running time guarantee in some regimes. See \cref{tab:Cp} for a comparison of results. 
\begin{table}[h]
\centering
\caption{Comparison with previous PCP/PCR runtimes. (See \cref{sec:results} for notation.)}
\label{tab:Cp}
\begin{tabular}{ c | c }
 \toprule
 {\bf Algorithm} & {\bf Runtime} \\
  \midrule\midrule
 FMMS16 \cite{FMMS16} & $\tilde{O}\left(\frac{1}{\gamma^2}\left(\nnz(\A)+d\cdot\mathrm{sr}(\A)\kappa\right)\right)$ \\ 
 \midrule
 AL17 \cite{AL17}, MMS18 \cite{musco2018stability}  & $\tilde{O}\left(\frac{1}{\gamma}\left(\nnz(\A)+d\cdot\mathrm{sr}(\A)\kappa\right)\right)$ \\
 \midrule
 \cref{thm:pcp_main,thm:pcr_main} & $\tilde{O}\left(\nnz(\A)+\sqrt{\nnz(\A)\cdot d\cdot\mathrm{sr}(\A)}\kappa/\gamma\right)$\\
 \bottomrule
\end{tabular}
\end{table}

\textbf{Asymmetric SVRG and Iterative Methods for Solving Linear Systems}: Variance reduction or varianced reduced iterative methods (e.g. SVRG \cite{johnson2013accelerating} is a powerful tool for improving convergence of stochastic methods. There has been work that used SVRG to develop primal-dual algorithms for solving saddle-point problems and extended it to monotone operators~\cite{pal16}. Our asymmetric SVRG solver can be viewed as an instance of their algorithm. We obtain improved running time analysis by performing a more fine-grained analysis exploiting problem structure.
For solving non-PSD system $(\A^\top\A-c\I)\xx=\vv$ satisfying $(\A^\top\A-c\I)^2\succeq\mu^2\I,\mu>0$, we provide~\cref{tab:Cp-2} to comparing the effectiveness of our asymmetric SVRG solver with some classic optimization methods for solving the problem (full discussion in~\cref{sec:SVRG} and~\cref{App:direct}). 
Note that when the condition number (i.e. $(\lambda_1/\mu)^2$) of the squared PSD system ($(\A^\top\A-c\I)^2$) is sufficiently large, our method (see~\cref{cor:square_solver_main}) can yield substantially improved running times.

\begin{table}[h]
\centering
\caption{Comparison for runtimes of solving non-PSD system $(\A^\top\A-c\I)\xx=\vv$.}
\label{tab:Cp-2}
\begin{tabular}{ c | c }
 \toprule
 {\bf Method} & {\bf Runtime} \\
  \midrule\midrule
AGD applied to squared counterpart & $\tilde{O}(\nnz(\A)\lambda_1/\mu)$ \\ 
 \midrule
SVRG applied to squared counterpart & $\tilde{O}(\nnz(\A)+\nnz(\A)^{3/4}d^{1/4}\sr(\A)^{1/2}\lambda_1/\mu)$ \\
 \midrule
Asymmetric SVRG (\cref{cor:square_solver_main}) & $\tilde{O}\bigl(\nnz(\A)+\sqrt{\nnz(\A)d\cdot\mathrm{sr}(\A)}\lambda_1/\mu\bigr)$\\
 \bottomrule
\end{tabular}
\end{table}

\subsection{Paper Organization}

The remainder of the paper is organized as follows. In \cref{sec:PCP}, we reduce the PCP problem\footnote{We refer reader to \cref{ssec:PCR} for the known reduction from PCR to PCP.} to matrix sign approximation and study the property of Zolotarev rational function used in approximation. In \cref{sec:SVRG}, we develop the asymmetric and squared linear system solvers using variance reduction and show the theoretical guarantee to prove \cref{thm:square_solver_main}, and correspondingly \cref{cor:square_solver_main}. In \cref{App:main}, we give the pseudocode of algorithms and formal proofs of main results in the paper, namely \cref{thm:pcp_main,thm:pcr_main,thm:square_root_solver_main}. In \cref{sec:experiment}, we conduct experiments and compare with previous methods to show efficacy of proposed algorithms. We conclude the paper in \cref{sec:conclusion}.

\section{PCP through Matrix Sign Approximation}
\label{sec:PCP}

Here we provide our reductions from PCP
to sign function approximation. We consider the rational approximation $r(x)$ found by~\citet{Z77} and study its properties for efficient (\cref{cor:approx}) and stable (\cref{lem:bound}) algorithm design to reduce the problem to solving squared ridge regressions. 

Throughout the section, we denote the sign function as $\sgn(x):\R\rightarrow\R$, where $\sgn(x) = 1$ whenever $ x> 0 $, $\sgn(x) = -1$ whenever $ x< 0 $, and $\sgn(0) = 0$. We also let $\mathcal{P}_k\defeq\{a_kx^k+\cdots+a_1x+a_0|a_k\neq0\}$ denote the class of degree-$k$ polynomials and $\mathcal{R}_{m,n}\defeq\{r_{m,n}|r_{m,n}=p_m/q_n,p_m\in\mathcal{P}_m$, $q_n\in\mathcal{P}_n\}$ denote the class of $(m,n)$-degree (or referred to as $\max\{m,n\}$-degree) rational functions.

For the PCP problem (see \cref{dfn:PCP}), we need an efficient algorithm that can approximately apply $\PP_\lambda$ to any given vector $\vv\in\R^d$. Consider the shifted matrix $\A^\top\A-\lambda\I$ so that its eigenvalues are shifted to $[-1,1]$ with $\lambda$ mapping to 0. Previous work has shown~\cite{FMMS16,AL17} solving PCP can be reduced to finding $f(x)$ that approximates sign function $\sgn(x)$ on $[-1,1]$, formally through the following reduction.

\begin{lemma}[Reduction: from PCP to Matrix Sign Approximation]\label{red:sign}
	Given a function $f(x)$ that $2\epsilon$-approximates
	 $\sgn(x)$:
	\begin{align}
		|f(x)-\sgn(x)|\le2\eps,\forall|x|\in[\lambda\gamma,1]
	\quad\text{and}\quad
	|f(x)|\le1,\forall x\in[-1,1],
	\label{cond:red}
	\end{align}
	then $\widetilde{\vv}=\frac{1}{2}\left(f(\A^\top\A-\lambda\I)+\I\right)\vv$ is an $\eps$-approximate PCP solution satisfying~(\ref{cond:PCP}).
	\end{lemma}
	
\begin{proof}%

Note that $\A^\top\A=\V\LLambda\V^\top$ where $\LLambda=\diag(\lambda_1,\cdots,\lambda_d)$ and each column of $\V$ is $\nnu_i,i\in[d]$. We can write $\vv= \sum_{i\in[d]}{\alpha_i\nnu_i}$ %
, and therefore $\PP_\lambda\vv=\sum_{i\in[k]} \alpha_i\nnu_i$. This also implies $f(\A^\top\A-\lambda\I)=\V f(\LLambda-\lambda\I)\V^\top$ where $f(\LLambda-\lambda\I) \defeq \diag(f(\lambda_1-\lambda),\cdots,f(\lambda_d-\lambda))$.

Now we define $k_1, k_2 \in [n]$ to divide the eigenvalues $\lambda_i$ and corresponding $\nnu_i$ into three settings, (1) $\lambda_i\ge\lambda(1+\gamma)$ for all $i\le k_1$, (2) $\lambda_i\in(\lambda(1-\gamma),\lambda(1+\gamma))$ for all $k_1<i\le k_2$, and (3) $\lambda_i\le\lambda(1-\gamma)$ for all $i > k_2$. Since by assumption $0<\lambda<\lambda_1\in[1/2,1]$ and $\gamma\in(0,1)$, it holds that  $\lambda_i-\lambda\in[\lambda\gamma,1]$
when $i\le k_1$)%
. Similarly, $\lambda_i-\lambda\in[-1,-\lambda\gamma]$ when $i< k_2$ and $\lambda_i-\lambda\in(-\lambda\gamma,\lambda\gamma)$ when $k_1<i\le k_2$. Consequently,  $|f(\lambda_i-\lambda)-\sgn(\lambda_i-\lambda)|\le 2\eps,\forall \lambda_i\notin(\lambda(1-\gamma),\lambda(1+\gamma))$ and we have that $|f(\lambda_i-\lambda)-\sgn(\lambda_i-\lambda)|\le 2,\forall \lambda_i\in(\lambda(1-\gamma),\lambda(1+\gamma))$.

 Noticing $\widetilde{\vv} = \frac{1}{2}\V(f(\LLambda-\lambda\I)+\I)\V^\top\vv$, and thus we have
 \begin{align*}
 \widetilde{\vv} - \vv  = \frac{1}{2}\V(f(\LLambda-\lambda\I)+\I)\V^\top\vv-\V\V^\top\vv
  = \frac{1}{2}\V(f(\LLambda-\lambda\I)-\I)\V^\top\vv 
 \end{align*}
 The result then follows by the following, which shows that  all conditions in~(\ref{cond:PCP}) are satisfied 
 \begin{align*}
 \|\PP_{(1+\gamma)\lambda}(\widetilde{\vv}-\vv)\|
  & \le\|\eps\V[\alpha_1,\cdots,\alpha_{k_1},0,\cdots,0]^\top\|\le\eps\|\vv\|,\\
 \|(\I-\PP_{(1-\gamma)\lambda})\tilde{\vv})\|
 & \le \|\eps\V[0,\cdots,0,\alpha_{k_2+1},\cdots,\alpha_d]^\top\|\le\eps\|\vv\|,\\
 \|(\PP_{(1+\gamma)} - \PP_{(1-\gamma)\lambda}) (\widetilde{\vv} - \vv)
\|
& \le
\|\V[\cdots,0,\alpha_{k_1+1},\cdots,\alpha_{k_2},0,\cdots]^\top\| =
\|(\PP_{(1+\gamma)} - \PP_{(1-\gamma)\lambda}) \vv \|.
 \end{align*}
 \end{proof}

In previous work~\cite{FMMS16,AL17}, the function $f(x)$ used to approximate $\sgn(x)$ are polynomials applied to the shifted-and-rescaled matrix $x\leftarrow(\A^\top\A+\lambda\I)^{-1}(\A^\top\A-\lambda\I)$.  Here the shifted-and-rescaled matrix $(\A^\top\A+\lambda\I)^{-1}(\A^\top\A-\lambda\I)$ is used instead of simply considering $\A^\top\A-\lambda\I$ to alleviate a multiplicative dependence on $1/\lambda$ in the runtime. These works noted that the optimal degree for achieving $|f(x)-\sgn(x)|\le2\eps,\forall|x|\in[\gamma,1]$
 is known by~\citet{eremenko2007uniform} to be $\tilde{O}(1/\gamma)$. Consequently, solving PCP requires solving at least $\tilde{O}(1/\gamma)$ ridge regressions in this framework.
 
 Departing from such an approach, we use Zolotarev rational function to directly approximate $\sgn( \A^\top\A-\lambda\I )$. This reduces the rational  degree to $O(\log(1/\lambda\gamma)\log(1/\eps))$, leading to the nearly input sparsity runtime improvement in the paper. 

Formally, Zolotarev rationals are defined as the optimal solution $r_k^\gamma(x)=x\cdot p(x^2)/q(x^2)\in\mathcal{R}_{2k+1,2k}$ for the optimization problem:
\begin{equation}
\begin{aligned}
& \underset{p,q\in P_k}{\text{max}}\underset{\gamma\le x\le1}{\text{min}}
& & x\frac{p(x^2)}{q(x^2)} \\
& \text{s.t.}
& & x\frac{p(x^2)}{q(x^2)}\le 1,\forall x\in[0,1]
\end{aligned}
\end{equation}
Zolotarev \cite{Z77} showed this optimization problem (up to scaling) is equivalent to solving
\begin{align*}
    \min\limits_{r\in \mathcal{R}_{2k+1,2k}}\max\limits_{|x|\in[\gamma,1]} |\mathrm{sgn}(x)-r(x)|
    ~.
\end{align*}
Further Zolotarev \cite{NF16}  showed that the analytic formula of $r^\gamma_k$ is given by 
\begin{align}
    r^\gamma_k(x) = C x\prod_{i\in[k]}\frac{x^2+c_{2i}}{x^2+c_{2i-1}} \label{eqn:Zolo}
    \text{ with }
c_i \defeq \gamma^2\frac{\mathrm{sn}^2(\frac{i K'}{2k+1};\gamma')}{\mathrm{cn}^2(\frac{i K'}{2k+1};\gamma')}, i\in[2k].
\end{align}
Here, all the constants depend on the explicit range $|x|\in[\gamma,1]$ we want to approximate uniformly. $C$ is computed through solving $ 1-r^\gamma_k(\gamma) = -(1-r^\gamma_k(1))$ %
, and coefficients $\{c_i\}_{i=1}^{2k}$ are computed through Jacobi elliptic coefficients, all as follows:
\begin{align*}
\text{coefficients} & 
\begin{cases}
	C & \defeq \frac{2}{(\gamma\prod_{i\in[k]}\frac{\gamma^2+c_{2i}}{\gamma^2+c_{2i-1}})+(\prod_{i\in[k]}\frac{1+c_{2i}}{1+c_{2i-1}})},\\
c_i & \defeq \gamma^2\frac{\mathrm{sn}^2(\frac{i K'}{2k+1};\gamma')}{\mathrm{cn}^2(\frac{i K'}{2k+1};\gamma')}, \forall i \in \{1, 2, \cdots,2k \}.
\end{cases}\\
\text{numerical constants} & 
\begin{cases}
	\gamma' & \defeq \sqrt{1-\gamma^2},\\
     K' & \defeq \int_0^{\pi/2}\frac{d\theta}{\sqrt{1-\gamma'^2\mathrm{sin}^2\theta}},\\
     u & \defeq F(\phi;\gamma')\defeq\int_0^\phi\frac{d\theta}{\sqrt{1-\gamma'^2\sin^2\theta}},\\ 
     \mathrm{sn}(u;\gamma') & \defeq\sin(F^{-1}(u;\gamma'));\ \mathrm{cn}(u;\gamma')\defeq\cos(F^{-1}(u;\gamma')).
\end{cases}
\end{align*}

This rational polynomial approximates $\sgn(x)$ on range $|x|\in[\gamma,1]$ with error decaying exponentially with degree, as formally characterized by the following theorem.

\begin{theorem}[Rational Approximation Error]
\label{cor:approx}
For any given $ \eps\in(0,1)$, when $k\ge\Omega(\log(1/\eps)\log(1/\gamma))$, it holds that
$$
\max_{|x|\in[\gamma,1]} |\mathrm{sgn}(x)-r^\gamma_k(x)| \le2\epsilon.$$
\end{theorem}

To prove \cref{cor:approx}, we first state the following classic result about the approximation error of $r_k^\gamma$ in \cref{lem:approx}. 

\begin{lemma}[Approximation Error]
\label{lem:approx}
The approximation error of $r^\gamma_k$ satisfies:
$$
\max\limits_{|x|\in[\gamma,1]} |\mathrm{sgn}(x)-r^\gamma_k(x)| = C_k \rho^{-2k+1}
\text{ for some  } 
C_{k}\in \left[\frac{2}{1+\rho^{-(2k+1)}},\frac{2}{1-\rho^{-(2k+1)}}\right]\subseteq [2,\frac{2}{1-\rho^{-1}}] ~.
$$
where $\rho\defeq\mathrm{exp}(\frac{\pi K(\mu')}{4K(\mu)})>1$,$K(\mu)\defeq\int_0^1\frac{dt}{\sqrt{(1-t^2)(1-\mu^2t^2)}}$, $\mu\defeq\frac{1-\sqrt{\gamma}}{1+\sqrt{\gamma}},\mu'\defeq\sqrt{1-\mu^2}$.
\end{lemma}

This lemma is simply a restatement of equation (32) in~\citet{gonchar1969zolotarev}. We refer interested readers to detailed derivation and proof there. 

It is crucial to our derivation of \cref{cor:approx} 
and a stable and efficient PCP algorithm that we bound the coefficients in $r^\gamma_k$ and $\rho$ in \cref{lem:approx}. The following lemma provides the key bounds we use for this purpose.%

\begin{lemma}[Bounding Coefficients]
\label{lem:bound}
The coefficients defined above have the following order / bounds (all constant are independent of $\lambda,\gamma,k$ and any other problem parameters):\\
(1) There exists constant $ 0<\beta_1<\infty$, such that $K(\mu)\le \beta_1(\log(1/\gamma)+1)$,%
(2) Coefficients $c_i$ are nondecreasing in $i$, $\forall i\in[2k]$. Also, there exists 
some constants $\beta_2>0,\beta_3<\infty$, such that $c_1\ge\beta_2 \frac{\gamma^2}{k^2}$, $c_{2k}\le \beta_3 k^2, \forall i\in[2k]$.
\end{lemma}

The proof of \cref{lem:bound} can be found in \cref{App:prop}. Now we use this to prove \cref{cor:approx}.

\begin{proof}[Proof of \cref{cor:approx}]
We apply \cref{lem:approx} and simply need to show that $C_k \rho^{-2k+1} \leq 2\epsilon$. Since $\mu'\in[0,1]$ we have $K(\mu')\ge\int_0^1  (1-t^2)^{-1/2} dt = \pi/2$.  Therefore, by (1) in \cref{lem:bound} we have that for some constant $\beta > 0$ we have $\rho \geq \exp(\beta / (\log(1/\gamma)+1))$ and \cref{lem:approx} then yields that
\[
C_k \rho^{-2k+1} \leq \frac{2}{\rho^{2k + 1} - 1}
\leq \frac{2}{\exp(\beta(2k +1)) / (\log(1/\gamma)+1)) -1} ~.
\]
The result follows as $k\ge\Omega(\log(1/\eps) \log(1/\gamma))$.
\end{proof}

As a quick illustration,
\cref{fig:Zooo} shows a comparison between the approximation errors of Zolotarev rational function, the Taylor expansion based polynomial used in~\cite{FMMS16} and Chebyshev polynomial used in~\cite{AL17} all with same degree.

\begin{figure}[htb]%
\centering
\subfigure[$\gamma$=0.1]{%
\includegraphics[width=0.33\textwidth]{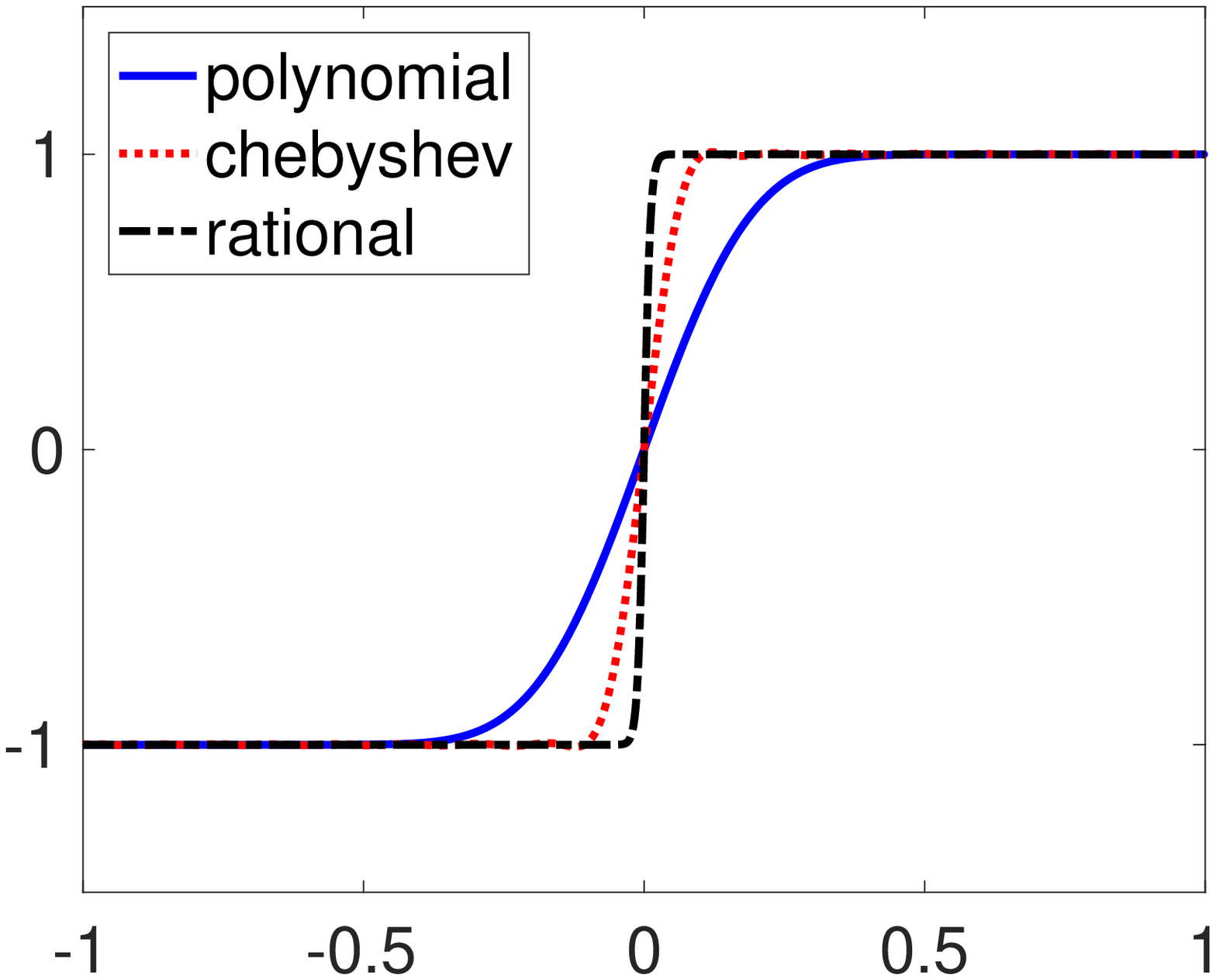}}%
\hfill
\subfigure[$\gamma$=0.05]{%
\includegraphics[width=0.33\textwidth]{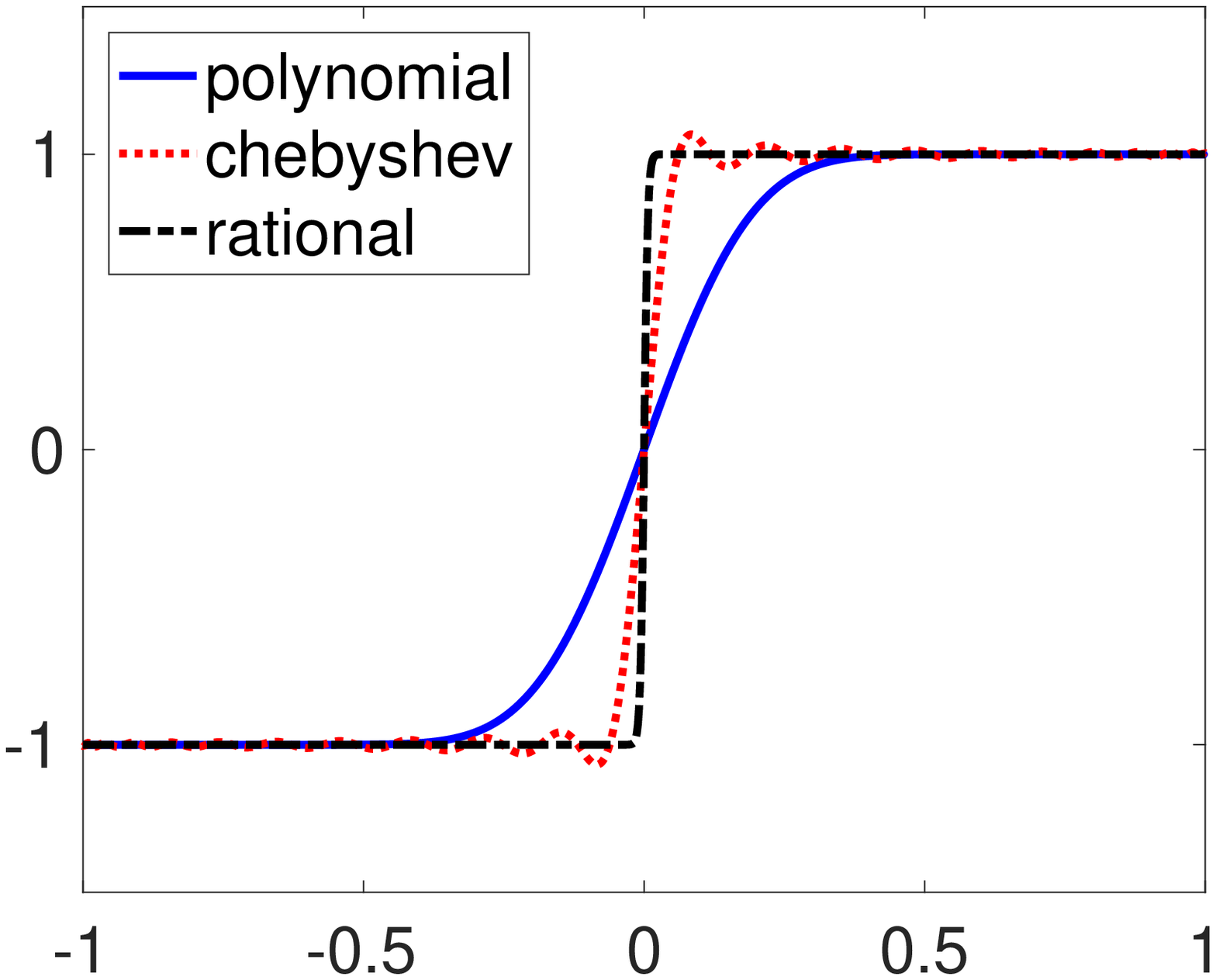}}%
\hfill
\subfigure[$\gamma$=0.01]{%
\includegraphics[width=0.33\textwidth]{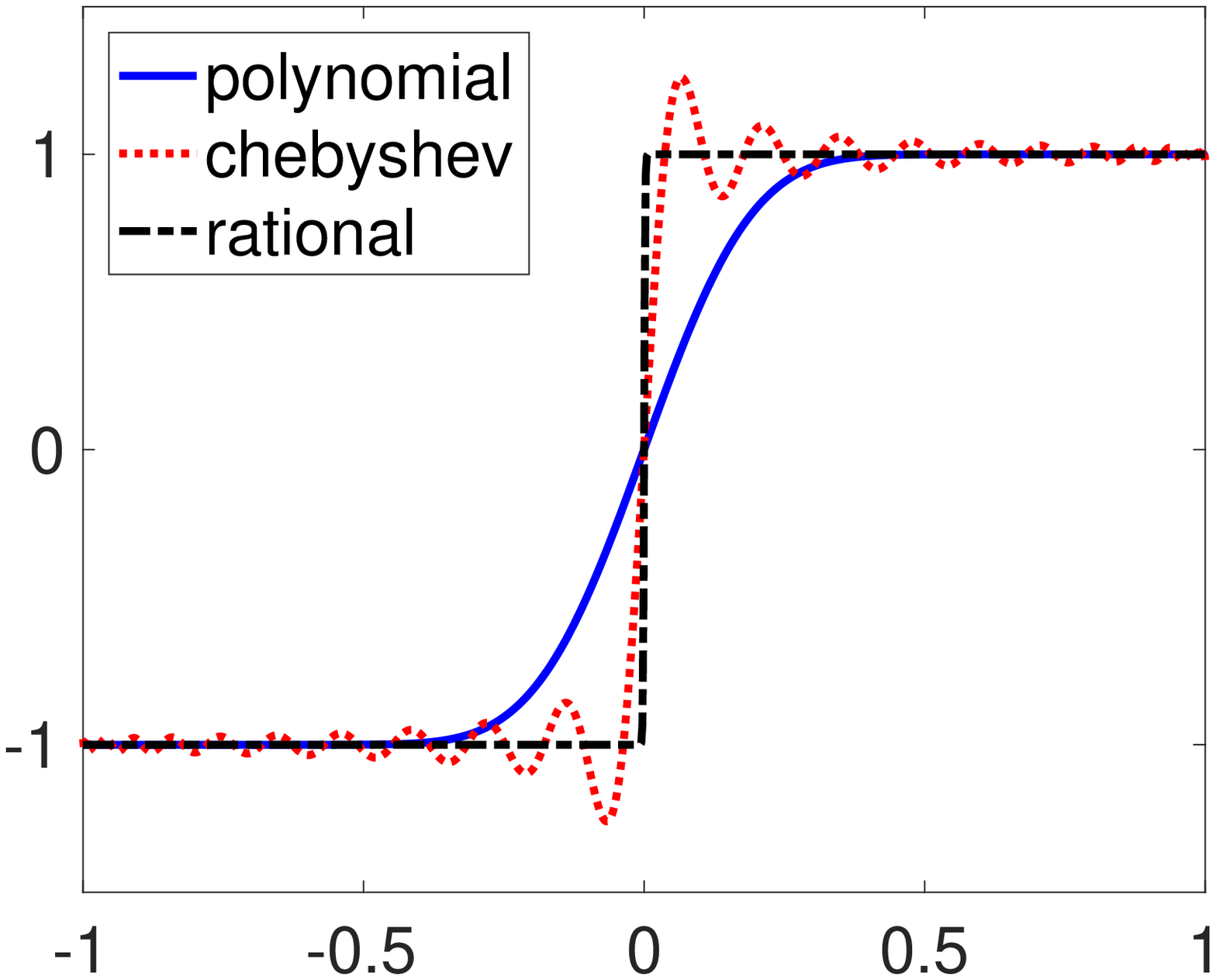}}
\caption{same degree = 21, different $\gamma$}
\label{fig:Zooo}
\vspace{-0.3cm}
\end{figure}
Treating $r^{\lambda\gamma}_k$ with $k\ge\Omega(\log(1/\eps)\log(1/\lambda\gamma))$ as the desired $f$ in \cref{red:sign}, it suffices to compute
\[ r^{\lambda\gamma}_k((\A^\top\A-\lambda\I))\vv= C (\A^\top\A-\lambda\I)\prod_{i=1}^k\frac{(\A^\top\A-\lambda\I)^2+c_{2i}\I}{(\A^\top\A-\lambda\I)^2+c_{2i-1}\I}\vv.
\]
To compute this formula approximately, we need to solve squared linear systems of the form $((\A^\top\A-\lambda\I)^2+c_{2i-1}\I)\xx=\vv$, the hardness of which is determined by the size of $c_{2i-1}(>0)$. The larger $c_{2i-1}$ is, the more positive-definite (PD) the system becomes, and the faster we can solve it. However, (2) in \cref{lem:bound} shows that, the $r_k^{\lambda\gamma}$ to use has coefficients $c_{i}=\tilde{\Omega}(1/\lambda^2\gamma^2)$ since $k=\Theta(\log(1/\eps)\log(1/\lambda\gamma))$. In the next sections we use this to bound the complexity of solving these squared systems.

\section{SVRG for Solving Asymmetric / Squared Systems}
\label{sec:SVRG}

In this section, we reduce solving squared systems to solving asymmetric systems (\cref{red:asySVRG}) and develop SVRG-inspired solvers (\cref{alg:AsySVRG}) to solve them efficiently. We study the theoretical guarantees of the resulting asymmetric SVRG method in both general settings (\cref{thm:asymmetric_main,thm:accel_asymmetric_main}) and our specific case (\cref{thm:main-asySVRG}).  Certain proofs in this section are deferred to \cref{App:SVRG}.

In \cref{sec:PCP}, we obtained low-degree rational function approximation $r^{\lambda\gamma}_k(x)$ of $\sgn(x)$ for $|x|\in[\lambda\gamma,1]$. In contrast to previous polynomial approximations to $\sgn(x)$, which \cite{FMMS16,AL17} used to reduce PCP to solving ridge-regression-type subproblems $(\A^\top\A+\lambda\I)\xx=\vv$ , our rational polynomial approximations reduce PCP to more complicated problem of solving the following squared systems:
\begin{align}
	[(\A^\top\A-c\I)^2+\mu^2\I]\xx=\vv,\text{ with }\A\in\R^{n\times d},\vv\in\R^d,\mu>0,c\in[0,\lambda_1].\label{eqn:square}
\end{align}

When the squared system is ill-conditioned (i.e. when $\lambda_1/\mu\gg0$), previous state-of-the-art methods can have fairly large running times. As shown in \cref{App:direct}, accelerated gradient descent~\cite{Nes83} applied to solving \eqref{eqn:square} has a runtime of $\tilde{O}(\nnz(\A)\lambda_1/\mu)$, which is not nearly linear in $\nnz(\A)$ in this regime. Applying the standard SVRG~\cite{johnson2013accelerating} techniques to the same system leads to a runtime $\tilde{O}(\nnz(\A)+d\cdot \mathrm{sr^2}(\A)\lambda_1^4/\mu^4)$, where $\mathrm{sr^2}(\A)\lambda_1^4/\mu^4$ comes from the high variance in sampling $\aaa_i\aaa_i^\top\aaa_j\aaa_j^\top$ from $(\A^\top\A)^2$ independently. Combining this with acceleration techniques \cite{allen2017katyusha,nesterov2017efficiency}, the best accelerated guarantee we known $\tilde{O}(\nnz(\A)+\nnz(\A)^{3/4}d^{1/4}\sr(\A)^{1/2}\lambda_1/\mu)$. (See~\cref{tab:Cp-2} for a comparison of the guarantee between these methods and our algorithm.)

To obtain more computationally efficient methods, rather than working with the squared system directly, we instead consider an equivalent formulation of the problem in a larger dimensional space. In this new formulation we show that we can develop lower variance stochastic estimators at the cost of working with an asymmetric (as opposed to symmetric) linear system which comes from the Schur complement decomposition of block matrices: %
\begin{equation}
\label{eq:SVRG-crude}
	\begin{aligned}
  \M  &\defeq \begin{pmatrix}
   \I & -\frac{1}{\mu}(\A^\top\A-c \I) \\
   \frac{1}{\mu}(\A^\top\A-c \I) & \I
  \end{pmatrix}\\
  & =
  \begin{pmatrix}
   \I & \0 \\
   \frac{1}{\mu}(\A^\top\A-c \I) & \I
  \end{pmatrix}
  \begin{pmatrix}
   \I & \0 \\
   \0 & \I+\frac{1}{\mu^2}(\A^\top\A-c \I)^2
  \end{pmatrix}
  \begin{pmatrix}
   \I & -\frac{1}{\mu}(\A^\top\A-c \I) \\
   \0 & \I
  \end{pmatrix}.
\end{aligned}
\end{equation}

\begin{lemma}[Reducing Squared Systems to Asymmetric Systems]
\label{red:asySVRG}
Define $\zz^*$ as the solution to the following asymmetric linear system:
	\begin{align}
	\begin{pmatrix}
   \I & -\frac{1}{\mu}(\A^\top\A-c \I) \\
   \frac{1}{\mu}(\A^\top\A-c \I) & \I
  \end{pmatrix}
  \zz
  =
    \begin{pmatrix}
   \0 \\
   \vv/\mu^2
  \end{pmatrix}.
  \label{eqn:aym}
  \end{align}
If we are given a solver that returns with probability $1-\delta$ a solution $\widetilde{\zz}$ satisfying 
  $\|\widetilde{\zz}-\zz^*\|_2\le\eps$ within runtime $\mathcal{T}(\eps,\delta)$,
  then we can use it to get an $\eps$-approximate squared ridge regression solver (see \cref{dfn:square}) with runtime $\mathcal{T}(\eps\|\vv\|,\delta)$ .
\end{lemma}

\begin{proof}%
By definition and the decomposition of $\M$ we have clearly $\M$ is invertible since both the upper and lower triangular matrix and the block diagonal matrix is invertible in~\eqref{eq:SVRG-crude}. Consequently, letting $\hat{\vv} \defeq [\0;\vv/\mu^2]$ we have that $\zz^* = \M^{-1}\hat{\vv}$ and therefore 
\begin{align*}
	\begin{pmatrix}
   \I & -\frac{1}{\mu}(\A^\top\A-c \I) \\
   \0 & \I
  \end{pmatrix}\zz^* & = \begin{pmatrix}
   \I^{-1} & \0 \\
   \0 & \left[\I+\frac{1}{\mu^2}(\A^\top\A-c \I)^2\right]^{-1}
  \end{pmatrix}
   \begin{pmatrix}
   \I & \0 \\
   -\frac{1}{\mu}(\A^\top\A-c \I) & \I
  \end{pmatrix}
\hat{\vv}\\
& = \begin{pmatrix}
   \0 \\
   \left[\I+\frac{1}{\mu^2}(\A^\top\A-c \I)^2\right]^{-1}\vv/\mu^2
  \end{pmatrix}
  ~.
\end{align*}
Taking the second half of the equations and write $\zz^*=[\xx^*;\yy^*]$ gives
 \begin{align*}
	\yy^* & =    \left[(\A^\top\A-c \I)^2+\mu^2\I\right]^{-1}\vv.
\end{align*}
This is to say optimal solution $\yy$ satisfies $\left((\A^\top \A-c \I)^2+\mu^2\I\right)\yy=\vv$.

As a result, if we have a solver that with high probability $\ge 1-\delta$ in time $\mathcal{T}(\eps',\delta)$ gives an $\eps'$-approximate solution $\widetilde{\zz}=[\widetilde{\xx},\widetilde{\yy}]$ of system $\M\zz=\hat{\vv}$ measured in $L_2$ norm, i.e. $\|\widetilde{\xx}-\xx^*\|_2^2+\|\widetilde{\yy}-\yy^*\|_2^2\le\eps'^2$, then as long as $\eps'^2\le\eps^2\|\vv\|^2$, we'll get 
\begin{align*}
\|\tilde{\yy}-\yy^*\|^2
\le \eps^2\|\vv\|^2,
\end{align*}
giving an $\eps$-approximation solution of squared ridge regression $\left((\A^\top\A-c\I)^2+\mu^2\I\right)\yy=\vv$ in time $\mathcal{T}(\eps\|\vv\|,\delta)$ with probability $1-\delta$.%
\end{proof}

Reducing solving the squared system to solving an asymmetric linear system $\M\zz=\hat{\vv}$ helps the development of our iterative methods as we can develop low-variance low-rank estimators. Though this advantage simply comes at the cost of asymmetry, this asymmetric matrix has the desirable property of having PSD symmetric part, meaning $\M^\top+\M\succeq2\I$. Formally, we define the asymmetric matrix $\M$ needed to develop our solver.

\begin{definition}
\label{def:PSD}
	We say that an  (asymmetric) matrix $\M$ is  \emph{$\mu$-positive-semidefinite ($\mu$-PSD)}, if 
	$\frac{1}{2}( \M^\top+\M )\succeq\mu\I$  for $\mu\ge0$.
	When the same condition holds for $\mu>0$, we also say it is \emph{$\mu$-strongly-positive-definite  ($\mu$-strongly-PD)}.\footnote{Usually PSD and PD matrices only refer to symmetric matrices. Here we abuse the notation through extending the definition to also asymmetric matrices by considering the symmetric part of the matrix $(\M^\top+\M)/2$.}.
\end{definition}

Those matrices enjoy the following properties:
	
\begin{lemma}[Properties of PSD matrices]
	\label{lem:PSD-prop}
	Here we list some properties that are used for 
	$\mu$-PSD $\M$.
	\begin{enumerate}
		\item Any $\mu$-strongly-PD $\M$ is invertible.
		\item For $\mu$-PSD $\M$, $\|(\I + \M)^{-1} \|_2 \leq \sqrt{\frac{1}{1+2\mu}}$. 
		\item For $\mu$-PSD $\M$, $\|(\tau \M^{-1} + \I)^{-1} \|_2 \leq 1$, $\forall \tau\ge0$. 
	\end{enumerate}
\end{lemma}

\begin{proof}
For the first property, we note that if $\exists\  \vv\neq0$ such that $\M\vv=0$, then 
\[
\vv^\top\M\vv=0,
\]
which contradicts with the condition that 
\[
\vv^\top\M\vv=\frac{1}{2}\vv^\top\left(\M^\top+\M\right)\vv\ge\mu\norm{\vv}^2>0,\forall \vv\neq\0.
\]

For the second property, note that 
\begin{equation}
\label{eq:upper1}
(\I + \M)^\top (\I + \M)
= \I + (\M + \M^\top) + \M^\top \M
\succeq (1 + 2\mu) \I 
\end{equation}
and therefore
\[
\|  (\I + \M)^{-1}\|^2
= \lambda_{\max} ( [(\I +  \M)^{-1}]^\top \I (\I + \M)^{-1})
\leq \frac{1}{1 + 2\mu} ~.
\]
where we used that $\I \preceq (1+2\mu)^{-1} (\I + \M)^\top (\I + \M)$ by \eqref{eq:upper1}. 

For the third property, when $\tau=0$ it obviously holds, so it suffices to prove for $\tau>0$. Since $\M$ is $\mu$-PSD, we have
\[
\M^\top \left([\M^{-1}]^\top +\M^{-1} \right) \M
= \M + \M^\top \succeq 2\mu \I
\]
multiplying on the left by $[\M^{-1}]^\top$ and the right $\M^{-1} $ yields that 
\[
[\M^{-1}]^\top +\M^{-1} \succeq 2\mu[\M^{-1}]^\top \M^{-1} \succeq \0
\]
Now given $\M/\tau$ is $\mu/\tau$-strongly-PD, as a result $\tau \M^{-1} +  [\tau \M^{-1}]^\top \succeq \0$ and using property 2 for $\tau\M^{-1}$ gives the result.
\end{proof}

Leveraging this notation, below we give the following formal definitions of the stochastic asymmetric problems that we derive efficient solvers for.
\begin{definition}[General Asymmetric Linear System Solver]
\label{def:general-asym}
Given $\hat{\vv}\in\R^{a}$ and a $\mu$-strongly-PD matrix $\M\in\R^{a\times a}$ such that $\M = \sum_{i\in[n]}\M_i$, $\|\M_i\|\le L_i,\forall i\in[n]$, the general asymmetric linear system is $\M\zz=\hat{\vv}$. An $\eps$-approximate general asymmetric solver returns (with high probability $1-\delta$ of some $\delta$ as input of the algorithm) an approximate solution $\widetilde{\zz}$ satisfying 
	\begin{align*}
\|\widetilde{\zz}-\M^{-1}\hat{\vv}\|\le\epsilon.
\end{align*} 
\end{definition}

We denote $\mathcal{T}_{\mathrm{mv}}(\M_i)$ as the cost to compute $\M_i \xx$ for an arbitrary vector $\xx$ and $\mathcal{T}=\max_{i\in[n]}\mathcal{T}_{\mathrm{mv}}(\M_i)$. Using SVRG methods as discussed in \cref{ssec:SVRG-gen}, we show how to give an $\eps$-approximate general asymmetric solver that runs in time $\tilde{O}(\nnz(\M)+\sqrt{\nnz(\M)\mathcal{T}}(\sum_{i\in[n]}L_i)/\mu)$.

For the particular case of interest as in~\eqref{eqn:aym}, correspondingly we have

\begin{definition}[Particular Asymmetric Linear System Solver]
\label{def:particular-asym}
Given $c\in[0,\lambda_1]$, $\hat{\vv}\in\R^{a}$ and a $1$-strongly-PD matrix $\M\in\R^{a\times a}$ in form
\begin{equation}
\label{cond:spec}
\M=\begin{pmatrix}
   \I & -\frac{1}{\mu}(\A^\top\A-c \I) \\
   \frac{1}{\mu}(\A^\top\A-c \I) & \I
  \end{pmatrix},
\end{equation}
 the particular asymmetric linear system is $\M\zz=\hat{\vv}$. An $\eps$-approximate particular asymmetric solver returns (with high probability $1-\delta$) an approximate solution $\widetilde{\zz}$ satisfying 
	\begin{align*}
\|\widetilde{\zz}-\M^{-1}\hat{\vv}\|\le\epsilon.
\end{align*} 
\end{definition}

Leveraging the general asymmetric solver with a more fine-grained analysis,  we show in \cref{ssec:SVRG-spec} a provable runtime of $\tilde{O}\bigl(\nnz(\A)+\sqrt{d\cdot\mathrm{sr}(\A)\nnz(\A)}\lambda_1/\mu\bigr)$ for the $\eps$-approximate particular asymmetric solver we build.

\subsection{SVRG for General Asymmetric Linear System Solving}
\label{ssec:SVRG-gen}

The general goal for this section is to build fast $\eps$-approximate asymmetric solver (see~\cref{def:general-asym}). All results in this subsection can be viewed as a variant of~\citet{pal16} and can be recovered by their slightly different algorithm which used proximal methods.

For solving the system, we consider the following update using the idea of variance-reduced sampling~\cite{johnson2013accelerating}: At step $t$, we conduct 
\begin{equation}
\begin{aligned}
\text{sample }i_t\in[n]\text{ independently with }p_{i_t}=L_{i_t}/(\sum_{i\in[n]}L_{i}),\\
\text{update }\zz_{t+1} := \zz_t-\frac{\eta}{p_{i_t}}\bigl(\M_{i_t} \zz_t-\M_{i_t}\zz_0+p_{i_t}(\M\zz_0-\hat{\vv})\bigr).\label{eqn:SVRG-one-step}
\end{aligned}
\end{equation}

For the above update, when variance is bounded in a desirable way, we can show the following expected convergence guarantee starting at a given initial point $\zz_0$:

\begin{lemma}[Progress per epoch]
\label{lem:gen-per-iter}
When the following variance bound on sampling holds 
$$\sum_{i\in[n]}\frac{1}{p_{i}}\|\M_i\zz_t-\M_i\zz_0 +p_{i}\cdot (\M\zz_0-\hat{\vv})\|^2\le2S^2\left[\|\zz_t-\zz^*\|^2+\|\zz_0-\zz^*\|^2\right],
$$ 
the sampling method~\eqref{eqn:SVRG-one-step} with a fixed step-size $\eta$ gives, after $T$ iterations, gets
$$
    \mathbb{E} \left\|\frac{1}{T}\sum_{t=0}^{T-1} \zz_{t}-\zz^* \right\|^2\le\frac{\frac{1}{2T}+\eta^2S^2}{\eta\mu-\eta^2S^2}\|\zz_0-\zz^*\|^2.
$$
\end{lemma}

\begin{proof}
For a single step at $t$ as in~\eqref{eqn:SVRG-one-step}, denote the index we drew as $i_t\in[n]$,
\begin{align*}
\|\zz_{t+1}-\zz^*\|^2= & \|\zz_t-\zz^*\|^2-2\frac{\eta}{p_{i_t}}\langle\M_{i_t}\zz_{t}-\M_{i_t}\zz_0 +p_{i_t}(\M\zz_0-\hat{\vv}),\zz_t-\zz^*\rangle\\
& +\frac{\eta^2}{p_{i_t}^2}\|\M_{i_t}\zz_{t}-\M_{i_t}\zz_0 +p_{i_t}(\M\zz_0-\hat{\vv}))\|^2.
\end{align*}
Taking expectation w.r.t $i_t$ we sample we get, %
\begin{equation}\label{eqn:SVRG-gen-sum}
\mathbb{E}_{i_t}\|\zz_{t+1}-\zz^*\|^2=\|\zz_t-\zz^*\|^2-2\eta\langle \M\zz_t -\hat{\vv},\zz_t-\zz^*\rangle+\eta^2\sum_{i\in[n]}\frac{1}{p_{i}}\|\M_i\zz_{t}-\M_i\zz_0 +p_i(\M\zz_0-\hat{\vv})\|^2.
\end{equation}
We bound the second and third terms on RHS respectively. For the second term we use the fact that $(\M+\M^\top)/2\succeq\mu\I$ by assumption and get
\begin{equation}
\begin{aligned}
    -\langle \M\zz_t-\hat{\vv} ,\zz_t-\zz^*\rangle
     = -\langle \M(\zz_t-\zz^*) ,\zz_t-\zz^*\rangle
     \le -\mu\|\zz_t-\zz^*\|^2.
\end{aligned}
\label{eqn:SVRG-gen-term2}
\end{equation}
For the third term using condition of bounded variance
\begin{equation}
\begin{aligned}
& \sum_{i\in[n]}\frac{1}{p_{i}}\|\M_i\zz_t-\M_i\zz_0 +p_{i}\cdot (\M\zz_0-\hat{\vv})\|^2
\le 2S^2 \left[\|\zz_t-\zz^*\|^2+\|\zz_0-\zz^*\|^2\right].	
\end{aligned}
\label{eqn:SVRG-gen-term3}	
\end{equation}
Combining~\eqref{eqn:SVRG-gen-sum},~\eqref{eqn:SVRG-gen-term2} and~\eqref{eqn:SVRG-gen-term3} we get,
\begin{align*}
\mathbb{E}_{i_t}\|\zz_{t+1}-\zz^*\|^2 & =\|\zz_t-\zz^*\|^2-2\eta\mu\|\zz_t-\zz^*\|^2+2\eta^2S^2[\|\zz_t-\zz^*\|^2+\|\zz_0-\zz^*\|^2]
\end{align*}
and equivalently,
\begin{align*}
(2\eta\mu-2\eta^2 S^2)\|\zz_t-\zz^*\|^2\le\|\zz_t-\zz^*\|^2-\mathbb{E}_{i_t}\|\zz_{t+1}-\zz^*\|^2+2\eta^2 S^2\|\zz_0-\zz^*\|^2.
\end{align*}
Taking expectation of $i_{t-1},\cdots,i_0$ respectively, averaging over $t=0,1,\cdots,T-1$ thus telescoping the first terms on RHS, and then rearranging terms, we get %
\begin{align*}
    \mathbb{E} \left\|\frac{1}{T}\sum\limits_{t=0}^{T-1} \zz_{t}-\zz^* \right\|^2\le\mathbb{E}\left[\frac{1}{T}\sum\limits_{t=0}^{T-1}\| \zz_{t}-\zz^*\|^2\right]\le\frac{\frac{1}{2T}+\eta^2S^2}{\eta\mu-\eta^2S^2}\|\zz_0-\zz^*\|^2.
\end{align*}

\end{proof}

In the general setting, the variance bound $S$ satisfies the following.
\begin{lemma}[Variance bound for general case]
\label{lem:var-gen}
When $\|\M_i\|\le L_{i},\forall i\in[n]$, the variance bound \eqref{eqn:SVRG-one-step} holds with $S=\sum_{i\in[n]} L_{i}$, i.e. 
$$\sum_{i\in[n]}\frac{1}{p_{i}}\|\M_i\zz_t-\M_i\zz_0 +p_{i}\cdot (\M\zz_0-\hat{\vv})\|^2
\le 2 \bigg(\sum_{i\in[n]}L_i\bigg)^2\left[\|\zz_t-\zz^*\|^2+\|\zz_t-\zz^*\|^2\right].$$
\end{lemma}
\begin{proof}
Note that
\begin{align*}
& \sum_{i\in[n]}\frac{1}{p_{i}}\|\M_i\zz_t-\M_i\zz_0 +p_{i}\cdot (\M\zz_0-\hat{\vv})\|^2 \\
& \stackrel{(i)}{=} \sum_{i\in[n]}\left(\frac{1}{p_{i}}\|(\M_i\zz_t-\M_i\zz^*)+(\M_i\zz^*-\M_i\zz_0) +p_{i}\cdot (\M\zz_0-\M\zz^*)\|^2\right)\\
& \stackrel{(ii)}{\le} \sum_{i\in[n]}\left(\frac{2}{p_{i}}\|\M_i(\zz_t-\zz^*)\|^2+2p_{i}\|\frac{1}{p_i}\M_i(\zz^*-\zz_0) + \M(\zz_0-\zz^*)\|^2\right)\\
& \stackrel{(iii)}{\le} \sum_{i\in[n]}\left(\frac{2}{p_{i}}\|\M_i(\zz_t-\zz^*)\|^2+\frac{2}{p_i}\|\M_i(\zz^*-\zz_0)\|^2\right)\\
& \le 2(\sum_{i\in[n]}L_i)^2 [\|\zz_t-\zz^*\|^2+\|\zz_0-\zz^*\|^2],
\end{align*}
where we use (i) $\M\zz^*=\hat{\vv}$ by linear system condition, (ii) $\|\zz+\zz'\|^2\le2\|\zz\|^2+2\|\zz'\|^2$, and (iii) $\mathbb{E}\|\zz-\mathbb{E}\zz\|^2\le\mathbb{E}\|\zz\|^2$.
\end{proof}

As a result, $\AsySVRG(\M,\hat{\vv},\zz_0,\eps,\delta)$ as shown in \cref{alg:AsySVRG} has the following guarantee. 

\begin{theorem}[General Asymmetric SVRG Solver]
\label{thm:asymmetric_main}
	For asymmetric system $\M\zz=\hat{\vv}$ (see~\cref{def:general-asym}), there is an $\eps$-approximate general asymmetric solver solver $\AsySVRG(\M,\hat{\vv},\zz_0,\eps,\delta)$ as specified in \cref{alg:AsySVRG} that runs in time 
	$\tilde{O}(\nnz(\M)+\mathcal{T}(\sum_{i\in[n]}L_i)^2/\mu^2)$.
\end{theorem}

\begin{proof}%
Note we choose $\eta = \mu/(4S^2)$ and $T=8S^2/\mu^2$ as specified in~\cref{alg:AsySVRG}. Using \cref{lem:gen-per-iter} we get
$$
    \mathbb{E}\norm{\frac{1}{T}\sum\limits_{t=0}^{T-1} \zz_{t}-\zz^*}^2\le\frac{2}{3}\|\zz_0-\zz^*\|^2.
$$
For the rest of the proof, we hide constants and write $\eta=\Theta(\mu/S^2)$ and $T=\Theta(S^2/\mu^2)$ to make a constant factor progress.

Leveraging this bound we can construct the asymmetric SVRG solver by using the above variance-reduced sampling Richardson iterative process repeatedly. After $T=\Theta(S^2/\mu^2)$ iterations, we update $\zz_0^{(q)}=\frac{1}{T}\sum_{t=0}^{T-1}\zz_t^{(q-1)}$ and repeat the process, initalized from this average point. Consequently, after %
\[
Q = \mathrm{log}_{3/2}(\|\zz_0-\zz^*\|^2/\epsilon)=O(\mathrm{log}(\|\zz_0-\zz^*\|/\epsilon))
\]
epochs, in expectation this algorithms returns a solution $\zz_0^{(Q+1)}=[\sum_{t=0}^{T-1}\zz_t^{(Q)}]/T$ satisfying
$$
\mathbb{E}\|\zz_0^{(Q+1)}-\zz^*\|^2\le\epsilon.
$$
For runtime analysis, notice within each epoch the computational cost is one full computation of $\M\zz$ as $\nnz(\M)$ and $T$ computations of $\M_i\zz$ upper bounded by $\mathcal{T}$ for each. As a result, letting $S=\sum_{i\in[n]}L_i$ the runtime in total to achieve $\|\zz-\zz^*\|\le\eps$ is 
$$
O\left(\left(\nnz(\M)+\mathcal{T}\cdot\frac{S^2}{\mu^2}\right)\mathrm{log}\left(\frac{\|\zz_0-\zz^*\|}{\epsilon}\right)\right).
$$

Replacing now $\eps$ by $\eps\delta$ and through Markov inequality argument that%
$$
\mathbb{P}(\|\zz_0^{(Q+1)}-\zz^*\|^2\ge\eps)\le \frac{\E[\|\zz_0^{(Q+1)}-\zz^*\|^2]}{\eps}\le\delta,
$$
 we transfer this algorithm to output the desired solution in $\eps$-approximation with probability $1-\delta$.
	
\end{proof}

\begin{algorithm}[H]
\DontPrintSemicolon
  
  \KwInput{$\M\in\R^{a\times a}$, $\hat{\vv}\in\R^a$, $\zz_0\in\R^a,\eps$ desired accuracy, $\delta$ probability parameter.}
  \KwOutput{$\zz_0^{(Q+1)}\in\R^a$.}
  Set $\eta=\mu/4(\sum_{i\in[n]}L_i)^2$,$T=\lceil(\sum_{i\in[n]}L_i)^2/\mu^2\rceil$, $p_i=L_i/(\sum_{i\in[n]}L_i),i\in[n]$ unless specified\;
  \For{$q=1$ to $Q=\Theta(\log(1/\eps\delta))$}
    {
        \For{$t\gets 1$ to $T$}
        {
        Sample $i_t\sim[n]$ according to $\{p_i\}_{i=1}^n$\;
        $\zz_t^{(q)}\gets \zz_{t-1}^{(q)}-\eta/p_{i_t}(\M_{i_t}\zz_{t-1}^{(q)}-\M_{i_t}\zz_0^{(q)}+p_{i_t}(\M \zz_0^{(q)}-\hat{\vv}))$\;
        }
        $\zz_0^{(q+1)} =\frac{1}{T}\sum_{t=1}^T \zz_t^{(q)}$\; 
    }
\caption{$\AsySVRG(\M,\hat{\vv},\zz_0,\eps,\delta)$}
\label{alg:AsySVRG}
\end{algorithm}

Now we turn to discuss the acceleration runtime. Inspired by approximate proximal point~\cite{FGKS15} or Catalyst~\cite{LMH15}, and using exactly the same acceleration technique used in~\citet{pal16}, when $\nnz(\M)\le\mathcal{T}(\sum_{i\in[n]}L_i)^2/\mu^2$, we can further improve this running time through the procedure as specified in \cref{alg:asy-acc}. 

\begin{algorithm}[H]
\DontPrintSemicolon
  
  \KwInput{$\M$, $\hat{\vv}$, $\zz_0$, $\eps$ desired accuracy, $\delta$ probability parameter}
  \Parameter{$\tau$, $I=\tilde{O}((\tau+\mu)/\mu)$}
  \KwOutput{$\zz^{I+1}$ as the final iterate from the outerloop}
  Initialize $\zz^{(1)}\gets \zz_0$
  \For{$i\gets 1$ to $I$}
  {
  $\zz^{(i+1)}\leftarrow \AsySVRG(\tau\I+\M,\tau\zz^{(i)}+\hat{\vv},\zz^{(i)},(\tfrac{\mu/2}{\mu+2\tau})^2\norm{\zz^{(i)}-(\tau\I+\M)^{-1}(\tau\zz^{(i)}+\hat{\vv})}^2,\delta/(I+1))$\;
  }
\caption{$\Asyacc(\M,\hat{\vv},\zz_0,\eps,\delta)$}
\label{alg:asy-acc}
\end{algorithm}

\begin{theorem}[Accelerated Asymmetric SVRG Solver]
	\label{thm:accel_asymmetric_main}
	When $\nnz(\M)\le\mathcal{T}(\sum_{i\in[n]}L_i)^2/\mu^2$, the accelerated algorithm $\Asyacc(\M,\hat{\vv},\zz_0,\eps,\delta)$ as specified in \cref{alg:asy-acc} is an $\eps$-approximate general asymmetric solver that runs in
	$\tilde{O}\bigl(\sqrt{\nnz(\M)\mathcal{T}}(\sum_{i\in[n]}L_i)/\mu\bigr)$.
\end{theorem}

To prove the accelerated rate, we first describe the progress per outerloop (from $z^{(i)}$ to $z^{(i+1)}$) in~\cref{alg:asy-acc}.
\begin{lemma}[Progress per Outerloop]
\label{lem:gen-per-iter-acc}
Let $\M \in \R^{a \times a}$ be $\mu$-strongly-PD. Further, suppose for some $\zz^{(0)}, \zz^{(1)}, \hat{\vv} \in \R^a$, $\tau \ge \mu$, $\epsilon \geq 0$, and $\zz^*_{\tau} \defeq (\tau \I + \M)^{-1} (\tau \zz^{(0)} + \hat{\vv})$ satisfy
\[
\|\zz^{(1)} - \zz^*_{\tau}\| \leq 
\epsilon 
\|\zz^{(0)} - \zz^*_{\tau}\| .
\]
Then for $\zz^* \defeq \M^{-1} \hat{\vv}$ we have
\[
\|\zz^{(1)} - \M^{-1} \hat{\vv}\| \leq \left(\frac{1}{1 + \mu/(2\tau)} + \epsilon\right) \|\zz^{(0)} - \M^{-1} \hat{\vv}\|
\]
\end{lemma}

\begin{proof}
From the definition of $\zz^*_{\tau}$ we have
\[
\zz^*_{\tau} - \M^{-1} \hat{\vv}
= 
(\tau \I + \M)^{-1} (\tau \zz^{(0)} + \hat{\vv}	
- 
(\tau \I + \M) \M^{-1} \hat{\vv}
)
=
\tau (\tau \I + \M)^{-1} (\zz^{0} - \M^{-1} \hat{\vv}) ~.
\]
Since $\tau(\tau \I + \M)^{-1} = (\I+ \tau^{-1} \M)$ and $\tfrac{1}{2}([\tau^{-1} \M] + [\tau^{-1} \M]^\top) \succeq (\mu/\tau) \I$ by the assumptions on $\M$ we have that $\|\tau(\tau \I + \M)^{-1}\|_2 \leq (1+2\mu/\tau)^{-1/2}$ by second property in \cref{lem:PSD-prop} and therefore 
\begin{equation}
\label{eq:error_1}
\| \zz^*_{\tau} - \M^{-1} \hat{\vv} \| \leq \sqrt{\frac{1}{1 + 2\mu/\tau}} \|\zz^{(0)} - \M^{-1} \hat{\vv}\|\le \frac{1}{1 + \mu/(2\tau)} \|\zz^{(0)} - \M^{-1} \hat{\vv}\| ~,
\end{equation}
where the last inequality follows from the fact $\sqrt{\frac{1}{1 + 2\mu/\tau}}\le \frac{1}{1 + \mu/(2\tau)}$ for all $\mu\le\tau$.
Further we have that
\[
\zz^{(0)} - \zz^{*}_{\tau}
= (\tau \I + \M)^{-1} ((\tau \I + \M) \zz^{0} - \tau \zz^{0} - \M \M^{-1} \hat{\vv}  )
= (\tau \M^{-1} + \I)^{-1} (\zz^{(0)} - \M^{-1} \hat{\vv}) ~.
\]
 by the third property of~\cref{lem:PSD-prop}  we have that $\|(\tau \M^{-1} + \I)^{-1} \|_2 \leq 1$ and  
\begin{equation}
\label{eq:error_2}
\| \zz^{(0)} - (\tau \I + \M)^{-1} (\tau \zz^{0} + \hat{\vv}) \| \leq  \|\zz^{(0)} - \M^{-1} \hat{\vv}\| ~.
\end{equation}
As $\|z^{(1)} - \M^{-1} \hat{v}\| \leq \|\zz^{(1)} -\zz^* _{\tau}\| + \|\zz^*_{\tau}+ \M^{-1} \hat{\vv}\|$ the result follows by \eqref{eq:error_1} and \eqref{eq:error_2}.
\end{proof}

\begin{proof}[Proof of \cref{thm:accel_asymmetric_main}]

The acceleration runtime can be achieved through a standard outer acceleration procedure:

Denote the whole optimizer $\zz^*$ satisfying $\M\zz^*=\hat{\vv}$, $S=\sum_{i\in[n]}L_i$ as usual. When $\nnz(\M)\le\mathcal{T}\cdot S^2/\mu^2$, we choose $\tau=S\sqrt{\mathcal{T}/\nnz(\M)}\ge\mu$.

Using a similar derivation as in~\cref{thm:asymmetric_main} we know after $T=\tilde{O}(\cdot(\nnz(\M)+\mathcal{T}(S+\tau)^2/(\mu+\tau)^2))$,  we have $\zz^{(i+1)}$ satisfying with probability at least $1-\delta/(I+1)$
\begin{align*}
\|\zz^{(i+1)}-(\tau \I+\M)^{-1}(\tau \zz^{(i)}+\hat{\vv})\|^2
& \le \left(\frac{\mu/2}{\mu+2\tau}\right)^2\|\zz^{(i)}-(\tau \I+\M)^{-1}(\tau \zz^{(i)}+\hat{\vv})\|^2.
\end{align*}
Starting the induction from $i=0$ and using~\cref{lem:gen-per-iter-acc} (since we have $\tau\ge\mu$) recursively, it implies with probability at least $1-(i+1)/(I+1)\cdot\delta$,
\begin{align*}
    & \|\zz^{(i+1)}-\zz^*\|\le \left(\frac{1}{1+\mu/2\tau}+\frac{\mu/2}{\mu+2\tau}\right)\|\zz^{(i)}-\zz^*\|,\\
    & \|\zz^{(i+1)}-\zz^*\|\le \left(\frac{1}{1+\mu/2\tau}+\frac{\mu/2}{\mu+2\tau}\right)^{i+1}\|\zz^{(0)}-\zz^*\|.
\end{align*}

Note as we choose $I=\tilde{O}((\tau+\mu)/\mu)$ we have with probability $1-\delta$ after $I$ outerloops,
$$
\|\zz^{(I+1)}-\zz^*\|\le\eps,
$$
which takes a total runtime of 
$$
\tilde{O}\left(\bigl(\frac{\mu+\tau}{\mu}\bigr)\bigl(\nnz(\M)+\mathcal{T}\frac{(S+\tau)^2}{(\mu+\tau)^2}\bigr)\right) = \tilde{O}\left(\sqrt{\nnz(\M)\mathcal{T}}\frac{\sum_{i\in[n]}L_i}{\mu}\right).
$$

\end{proof}

\subsection{Asymmetric Linear System Solving for Squared System Solver}
\label{ssec:SVRG-spec}

To solve the particular asymmetric system we consider the step
\begin{equation}
\begin{aligned}
\zz_{t+1} & = \zz_t-\frac{\eta}{p_i}\bigl(\M_i \zz_t-\M_i\zz_0+p_i(\M\zz_0-\hat{\vv})\bigr)\\
\text{ where }
\M_i & \defeq
  \begin{pmatrix}
   \frac{\|\aaa_i\|^2}{\|\A\|_\mathrm{F}^2} \I & -\frac{1}{\mu}(\aaa_i \aaa_i^\top -c\frac{\|\aaa_i\|^2}{\|\A\|_\mathrm{F}^2} \I) \\
   \frac{1}{\mu}(\aaa_i\aaa_i^\top -c\frac{\|\aaa_i\|^2}{\|\A\|_\mathrm{F}^2}) \I & \frac{\|\aaa_i\|^2}{\|\A\|_\mathrm{F}^2} \I
  \end{pmatrix},\ p_i\propto\|\aaa_i\|^2,\ \forall i\in[n].
\end{aligned}
\label{eqn:SVRG-one-step-spec}
\end{equation}

Through a more fine-grained analysis, $\AsySVRG(\M,\hat{\vv},\zz_0,\eps,\delta)$ with particular choices of $\eta$, $T$, $\M_i$, and $\{p_i\}_{i\in[n]}$, can have a better runtime guarantee and be accelerated using a similar idea as in the general case discussed in previous subsection. This is stated formally using the following variance bound in \cref{thm:main-asySVRG}.

\begin{lemma}[Variance bound for specific form]
\label{lem:var-spec}
For problem~\eqref{cond:spec}, the variance incurred in~\eqref{eqn:SVRG-one-step-spec} is bounded by $S=O(\norm{\A}_\mathrm{F}\sqrt{\lambda_1}/\mu)$, i.e. there exists constant $0<C<\infty$ that
$$\sum_{i\in[n]}\frac{1}{p_{i}}\|\M_i\zz_t-\M_i\zz_0 +p_{i}\cdot (\M\zz_0-\hat{\vv})\|^2\le\frac{C\norm{\A}_\mathrm{F}^2\lambda_1}{\mu^2}\left[\|\zz_t-\zz^*\|^2+\|\zz_t-\zz^*\|^2\right].$$
\end{lemma}
\begin{proof}
For arbitrary $\Delta\in\R ^{2d}$, set $p_i=\|\aaa_i\|^2/\|\A\|_\mathrm{F}^2$, we have
\begin{align*}
& \sum_{i\in[n]}\frac{1}{p_{i}}\|\M _{i} \Delta\|^2\\
& = \Delta^\top \biggl(\sum_{i\in[n]}\frac{1}{p_{i}}\M _{i}^\top \M _{i}\biggr) \Delta\\
& \stackrel{(i)}{=} \Delta^\top \left(\sum_{i\in[n]}\frac{1}{p_i}
  \begin{pmatrix}
    p_i^2\I+\frac{1}{\mu^2}(\aaa_i\aaa_i^\top -p_ic \I)^2
   & \0 \\
   \0 & 
   \frac{1}{\mu^2}(\aaa_i\aaa_i^\top -p_ic \I)^2+p_i^2 \I
  \end{pmatrix}
\right) \Delta\\
& \stackrel{(ii)}{=} \Delta^\top 
  \begin{pmatrix}
   \I+\frac{1}{\mu^2}(\|\A\|_\mathrm{F}^2\A^\top \A-2c \A^\top \A+c^2\I) & 
   \0 \\
   \0 & 
   \I+\frac{1}{\mu^2}(\|\A\|_\mathrm{F}^2\A^\top \A-2c\A^\top \A+c^2\I)
  \end{pmatrix}\Delta
  \\
& \stackrel{(iii)}{\le} (1+\frac{\|\A\|_\mathrm{F}^2\lambda_1}{\mu^2})\|\Delta\|^2\le \frac{C\|\A\|_\mathrm{F}^2\lambda_1}{\mu^2}\|\Delta\|^2,
\end{align*}
where we use (i) the specific form of $\M_i$ as in~\eqref{eqn:SVRG-one-step-spec}, (ii) specific choice of $p_i=\norm{\aaa_i}^2/\norm{\A}_\mathrm{F}^2$, (iii) the assumption that $\A^\top\A\preceq\lambda_1\I$ and $c\in[0,\lambda_1]$. %

Note then similar to proof of \cref{lem:var-gen}, we have
\begin{align*}
\sum_{i\in[n]}\frac{1}{p_{i}}\|\M _{i} (\zz_t-\zz_0) +p_{i}(\M \zz_0-\vv)\|^2 & \le \sum_{i\in[n]}\biggl(\frac{2}{p_{i}}\|\M _{i} (\zz_t-\zz^*)\|^2+\frac{2}{p_i}\|\M _{i}(\zz^*-\zz_0)\|^2\biggr)\\
& \le\frac{2C\|\A\|_\mathrm{F}^2\lambda_1}{\mu^2}(\|\zz_t-\zz^*\|^2+\|\zz_0-\zz^*\|^2).
\end{align*}

\end{proof}

\begin{theorem}[Particular Asymmetric SVRG Solver]
\label{thm:main-asySVRG}
Set $p_i=\|\aaa_i\|^2/\|\A\|_\mathrm{F}^2$, $\eta=\mu^2/2\lambda_1\|\A\|_\mathrm{F}^2$, $T=\lceil 2\|\A\|_\mathrm{F}^2 \lambda_1/\mu^2\rceil$ and
$$
\M_i:=
  \begin{pmatrix}
   \frac{\|\aaa_i\|^2}{\|\A\|_\mathrm{F}^2}I & -\frac{1}{\mu}(\aaa_i \aaa_i^\top -c\frac{\|\aaa_i\|^2}{\|\A\|_\mathrm{F}^2} \I) \\
   \frac{1}{\mu}(\aaa_i\aaa_i^\top -c\frac{\|\aaa_i\|^2}{\|\A\|_\mathrm{F}^2}) \I & \frac{\|\aaa_i\|^2}{\|\A\|_\mathrm{F}^2} \I
  \end{pmatrix},\ \forall i\in[n].
$$
Then $\AsySVRG(\M,\hat{\vv},\zz_0,\eps,\delta)$ as specified in \cref{alg:AsySVRG} becomes an $\eps$-approximate particular asymmetric solver that runs in $\tilde{O}\bigl(\nnz(\A)+d\cdot\mathrm{sr}(\A)\lambda_1^2/\mu^2\bigr).$ An accelerated variant of it has the improved runtime of $\tilde{O}\bigl(\lambda_1\sqrt{\nnz(\A)d\cdot\mathrm{sr}(\A)}/\mu\bigr)$ when $\nnz(\A)\le d\cdot\mathrm{sr}(\A)\lambda_1^2/\mu^2$.
\end{theorem}

\begin{proof}

Using \cref{lem:var-spec}, the conditions of \cref{lem:gen-per-iter} are satisfied with $S=\sqrt{C}\norm{\A}_\mathrm{F}\sqrt{\lambda_1}/\mu$. Consequently we have
\begin{equation}
\begin{aligned}
    \mathbb{E}\left\|\frac{1}{T}\sum\limits_{t=0}^{T-1} \zz_{t}-\zz^*\right\|^2\le\mathbb{E}\left[\frac{1}{T}\sum\limits_{t=0}^{T-1}\| \zz_{t}-\zz^*\|^2\right]\le\frac{\frac{1}{2T}+\eta^2C\|\A\|_\mathrm{F}^2\lambda_1/\mu^2}{\eta-\eta^2C\|\A\|_\mathrm{F}^2\lambda_1/\mu^2}\|\zz_0-\zz^*\|^2.
\end{aligned}
\label{eqn:var-progress-per-spec}	
\end{equation}

{\bf For the non-acceleration case:}
Note we choose $\eta = \mu^2/(C\|\A\|_\mathrm{F}^2\lambda_1)$, and $T=C\|\A\|_\mathrm{F}^2\lambda_1/\mu^2$ as in \cref{eqn:var-progress-per-spec} yields
$$
    \mathbb{E}\left\|\frac{1}{T}\sum\limits_{t=0}^{T-1} \zz_{t}-\zz^*\right\|^2\le\frac{2}{3}\|\zz_0-\zz^*\|^2.
$$

Hereinafter, we hide constants and write $\eta=\Theta(\mu^2/(\|\A\|_\mathrm{F}^2\lambda_1))$ and $T=O(\|A\|_\mathrm{F}^2\lambda_1/\mu^2)$ to make a constant factor progress.

Then similarly as in \cref{alg:AsySVRG} and \cref{thm:asymmetric_main}, we argue after $Q=O(\mathrm{log}(\|\zz_0-\zz^*\|/\epsilon)$ batches, in expectation we'll return a solution $\zz_0^{(Q+1)}=[\sum_{t=0}^{T-1}\zz_t^{(Q)}]/T$
$$
\mathbb{E}\|\zz_0^{(Q+1)}-\zz^*\|^2\le\epsilon.
$$
For runtime analysis, notice within each batch the computational cost is one full computation of $\M\zz$ and $O(T)$ computations of $\M_i\zz$, which together is $O(\mathrm{nnz}(\A)+d\cdot\|\A\|_\mathrm{F}^2\lambda_1/\mu^2)$. So the total runtime to achieve $\|\zz^{(Q+1)}-\zz^*\|\le\eps$ with probability $1-\delta$ is $$O\bigl((\nnz(\A)+d\cdot\|\A\|_\mathrm{F}^2\lambda_1/\mu^2)\mathrm{log}(\|\zz_0-\zz^*\|/\epsilon\delta)\bigr)=\tilde{O}(\nnz(\A)+d\cdot\mathrm{sr}(\A)\lambda_1^2/\mu^2).$$

{\bf For the acceleration case:} The standard technique of outer acceleration used in \cref{thm:accel_asymmetric_main} is applied to get a better runtime under this case, and is used to prove \cref{thm:square_solver_main}. 

\end{proof}

\subsection{Squared System Solver Using SVRG}

For the squared ridge regression solver, we first give its simple pseudocode using \cref{red:asySVRG} and $\AsySVRG$ for completeness.

\begin{algorithm}[H]
\DontPrintSemicolon
  \KwInput{$\A$ data matrix, $c\in[0,\lambda_1]$, $\mu>0$, $\vv$, $\xx_0$ initial, $\eps$ accuracy, $\delta$ probability}
  \KwOutput{$\widetilde{\xx}$ $\eps$-approximate solution as in \cref{dfn:square}.}
  Initialize $\zz^{0}$.\;
  Set 
	$\M =\begin{pmatrix}
   \I & -\frac{1}{\mu}(\A^\top\A-c \I) \\
   \frac{1}{\mu}(\A^\top\A-c \I) & \I
  \end{pmatrix}$ and $\hat{\vv}=\left(\0,\vv/\mu^2\right)^\top$.\;
  Call $[\xx,\yy]\leftarrow\AsySVRG(\M,\hat{\vv},\zz_0,\eps\|\vv\|,\delta)$.\;
  Return $\widetilde{\xx}\gets\yy$\;
  \caption{$\RidgeSquare(\A,c,\mu^2,\vv,\eps,\delta)$}
  \label{alg:rs}
\end{algorithm}

This is essentially a corollary of \cref{thm:main-asySVRG}.

\begin{proof}[Proof of \cref{thm:square_solver_main}]

	Set 
	$$\M =\begin{pmatrix}
   \I & -\frac{1}{\mu}(\A^\top\A-c \I) \\
   \frac{1}{\mu}(\A^\top\A-c \I) & \I
  \end{pmatrix}
  \text{ and }
  \hat{\vv}=\left(\0,\vv/\mu^2\right)^\top,~c\in[0,\lambda_1].$$ We know an $\eps$-approximate squared ridge regression solver would suffice to call $$[\xx,\yy]\leftarrow\AsySVRG(\M,\hat{\vv},\zz_0,\eps\|\vv\|,\delta).$$ once and set its output as $\yy$ through \cref{red:asySVRG}. This together with \cref{thm:main-asySVRG} gives us the total runtime of $\tilde{O}\left(\nnz(\A)+d\cdot\mathrm{sr}(\A)\lambda_1^2/\mu^2\right)$ unaccelerated and $\tilde{O}\left(\sqrt{\nnz(\A)d\cdot\mathrm{sr}(\A)}\lambda_1/\mu\right)$ accelerated (when $\nnz(\A)\le d\cdot\sr(\A)\lambda_1^2/\mu^2$). In short, we have the guaranteed running time within 
  $$\tilde{O}\left(\nnz(\A)+\sqrt{\nnz(\A)d\cdot\mathrm{sr}(\A)}\lambda_1/\mu\right).$$
	\end{proof}

\cref{thm:square_root_solver_main} also implies immediately a solver for non-PSD system in form: $(\A^\top\A-c\I)\xx=\vv$, $c\in[0,\lambda_1]$ with same runtime guarantee whenever all eigenvalues $\lambda_i-c$ of  $(\A^\top\A-c\I)$ satisfy $|\lambda_i-c|\ge\mu>0,\forall i$. This is done by considering solving $(\A^\top\A-c\I)^2\xx=(\A^\top\A-c\I)\vv$. We state this formally in the following corollary for completeness, which is equivalent as showing \cref{cor:square_solver_main}.

\begin{corollary}
	\label{cor:square_solver_main_full}
	Given $c\in[0,\lambda_1]$, and a non-PSD system $(\A^\top\A-c\I)\xx=\vv$ and an initial point $\xx_0$, for arbitrary $c\in\R$ satisfying $(\A^\top\A-c\I)^2\succeq\mu^2\I,\mu>0$, there is an algorithm that uses $\AsySVRG(\M,\hat{\vv},\zz_0,\Theta(\eps\|\vv\|),\delta)$ to return with probablity $1-\delta$ a solution $\widetilde{\xx}$ such that $\|\widetilde{\xx}-(\A^\top\A-c\I)^{-1}\vv\|\le\epsilon\|\vv\|$, within runtime 
	$\tilde{O}\bigl(\nnz(\A)+d\cdot\mathrm{sr}(\A)\lambda_1^2/\mu\bigr).$ The runtime can be accelerated to $\tilde{O}\bigl(\lambda_1\sqrt{\nnz(\A)d\cdot\mathrm{sr}(\A)}/\mu\bigr)$ when $\nnz(\A)\le d\cdot\mathrm{sr}(\A)\lambda_1^2/\mu^2$.
\end{corollary}

\begin{proof}
It suffices to show that we can solve $(\A^\top\A-c\I)^2\xx=(\A^\top\A-c\I)\vv$ to high accuracy within the desired runtime for any given $\vv$.

By \cref{thm:square_solver_main}, whenever $(\A^\top \A-c\I)^2\succeq\mu^2\I$, we can solve $\bigl((\A^\top \A-c\I)^2+\mu^2\I\bigr)\xx=(\A^\top\A-c\I)\vv$ to $\eps$-approximate accuracy within runtime $\tilde{O}(\nnz(\A)+\sqrt{\nnz(\A)d\cdot\mathrm{sr}(\A)}\lambda_1/\mu).$

Now we can consider preconditioning $(\A^\top \A-c\I)^2$ using $(\A^\top \A-c\I)^2+\mu^2\I$ by noticing that $(\A^\top\A-c\I)^2+\mu^2\I\approx_{1/2}(\A^\top\A-c\I)^2$ under the condition.
	
	As a result, it suffices to apply Richardson update 
	$$x^{(t+1)}\gets x^{(t)}-\eta\left[(\A^\top\A-c\I)^2+\mu^2\I\right]^{-1}\left((\A^\top\A-c\I)^2\xx^{(t)}-(\A^\top\A-c\I)\vv\right),$$
	with $\eta=1/2$. Since it satisfies $\frac{1}{2}\I\preceq\left[(\A^\top\A-c\I)^2+\mu^2\I\right]^{-1}(\A^\top\A-c\I)^2\preceq\I$, it achieves $\eps$ accuracy $\|x^{(T)}-(\A^\top\A-c\I)^{-2}\vv\|\le\eps\|\vv\|$ in $$O\left(\log\left(\frac{\|\xx_0-(\A^\top\A-c\I)^{-1}\vv\|}{\eps\|\vv\|}\right)\right)=\tilde{O}(1)$$ iterations, with each iteration cost $\tilde{O}\bigl(\nnz(\A)+d\cdot\mathrm{sr}(\A)\lambda_1^2/\mu^2\bigr)$ using the unaccelerated subroutine, and $\tilde{O}\bigl(\sqrt{d\cdot\mathrm{sr}(\A)}\lambda_1/\mu\bigr)$ using the accelerated subroutine when $\nnz(\A)\le d\cdot\mathrm{sr}(\A)\lambda_1^2/\mu^2$.  This leads to a total unaccelerated runtime of $\tilde{O}\bigl(\nnz(\A)+d\cdot\mathrm{sr}(\A)\lambda_1^2/\mu\bigr)$, and accelerated runtime $\tilde{O}\bigl(\lambda_1\sqrt{\nnz(\A)d\cdot\mathrm{sr}(\A)}/\mu\bigr)$ when $\nnz(\A)\le d\cdot\mathrm{sr}(\A)\lambda_1^2/\mu^2$.
\end{proof}

In the end of this section, we remark that all proofs and results are stated without the condition $\lambda_1\le1$ to give a clearer sense of the runtime dependence in general setting. When applied to our specific case of solving squared systems, due to renormalization we have $\lambda_1\in[1/2,1]$ and $\lambda\leftarrow \lambda/\lambda_1=1/\kappa$ which lead to running times formally stated in \cref{thm:pcp_main,thm:pcr_main,thm:square_solver_main} and proved in \cref{App:main}.

\section{Nearly Linear Time PCP and PCR Solvers}
\label{App:main}

In this section, we prove our main theorems for new algorithms on PCP and PCR problems stated in \cref{sec:results}. As a byproduct of the results, we can also use some variant of the Zolotarev rational to approximate the square-root function, and thus build efficient square-root-matrix-and-vector solver (see~\cref{thm:square_root_solver_main} and~\cref{alg:SR}). Since the result is proved in a very similar way, we defer a full discussion to \cref{App:main-sup}.

To begin with, we introduce a helper lemma useful for analyzing the approximation property of the theorem and defer its proof to \cref{App:main-sup} for detail.

\begin{lemma}[Accumulative Error from Products]
\label{lem:stable}
	If there are procedures $\mathcal{C}_i(\vv),i\in[k]$ that carries out a product computation $\CC_i\cdot \vv$ in $\eps$ accuracy, i.e. $\forall i\in[k],\|\mathcal{C}_i(\vv)-\CC_i\vv\|\le\eps\|\vv\|$, and that $\|\CC_i\|\le M,\forall i\in[k]$ for some $M\ge1$. When $\eps\le M/2k$, we have
	$$
	\|\mathcal{C}_{k}(\mathcal{C}_{k-1}(\cdots\mathcal{C}_1(\vv)))-\prod\limits_{i=1}^k\CC_i\vv\|\le 2\eps kM^{k-1}\|\vv\|.
	$$
\end{lemma}

\subsection{PCP Solver}
\label{ssec:PCP_main}
Given
a squared ridge regression solver $\RidgeSquare(\A,\lambda,c_{2i-1},\vv,\eps,\delta)$ (see \cref{sec:SVRG}), using the reduction in \cref{sec:PCP} we can get an $\eps$-approximate PCP algorithm $\ISPCP(\A,\vv,\lambda,\gamma,\eps,\delta)$ shown in \cref{alg:ISPCP} and its theoretical guarantee in \cref{thm:pcp_main}.

\begin{algorithm}[H]
\DontPrintSemicolon
  \KwInput{$\A$ data matrix, $\vv$ projecting vector, $\lambda$ threshold, $\gamma$ eigengap, $\eps$ accuracy, $\delta$ probability.}
  \Parameter{degree $k$ (\cref{cor:approx}), coefficients $\{c_{i}\}_{i=1}^{2k},C$ (\cref{eqn:Zolo}), accuracy $\eps_1$ (specified below)}
  \KwOutput{A vector $\widetilde{\vv}$ that solves PCP $\eps$-approximately.}
	$\widetilde{\vv}\gets\vv$\;
  \For{$i\gets 1$ to $k$}
  {
  $\widetilde{\vv}\gets (\A^\top\A-\lambda \I)^2\widetilde{\vv}+c_{2i}\widetilde{\vv}$\;
  $\widetilde{\vv}\gets \RidgeSquare(\A,\lambda,c_{2i-1},\widetilde{\vv},\eps_1,\delta/k)$\;
  }
  $\widetilde{\vv}\gets C(\A^\top\A-\lambda\I)\widetilde{\vv}$, $\widetilde{\vv}\gets\frac{1}{2}(\vv+\widetilde{\vv})$.\;
\caption{$\ISPCP(\A,\vv,\lambda,\gamma,\eps,\delta)$}
\label{alg:ISPCP}
\end{algorithm}

\begin{proof}[Proof of \cref{thm:pcp_main}]
\ 

\paragraph{Choice of parameters:}

We choose the following values for parameters in \cref{alg:ISPCP}:
	\begin{align*}
	k & =\Omega(\log(1/\eps)\log(1/\lambda\gamma))\\
	M & =\beta_3k^4/\beta_2\gamma^2\lambda^2\\
	\eps_1 & = \frac{\eps}{8\beta_3k^3M^{k-1}}.
	\end{align*}
    The other coefficients $\{c_{i}\}_{i=1}^{2k},C$ are as defined in \cref{eqn:Zolo}. Further we use constants $\beta_2,\beta_3$ as stated in \cref{lem:bound}.
	
\paragraph{Approximation:}
	
	Given $\lambda>0,\gamma\in(0,1)$, from \cref{cor:approx} and the definition of $k$ and $r_k^{\lambda\gamma}(x)$ we get that $\max_{|x|\in[\lambda\gamma,1]}|\sgn(x)-r_k^{\lambda\gamma}(x)|\le2\eps$.

	Using \cref{red:sign}, we know for such $r_k^{\lambda\gamma}$, $\widetilde{\vv}=\frac{1}{2}(r_k^{\lambda\gamma}(\A^\top\A-\lambda\I)+\I)\vv$ satisfies the conditions in~(\ref{cond:PCP}), i.e.
\begin{align}
 & 1. \|\PP_{(1+\gamma)\lambda}(\widetilde{\vv}-\vv)\|\le\eps/2\|\vv\|;\nonumber\\
 & 2.	\|(\I-\PP_{(1-\gamma)\lambda})\widetilde{\vv}\|\le\eps/2\|\vv\|;\nonumber\\
& 3. 
\|(\PP_{(1+\gamma)} - \PP_{(1-\gamma)\lambda}) (\widetilde{\vv} - \vv)\|\le \|(\PP_{(1+\gamma)} - \PP_{(1-\gamma)\lambda}) \vv\|\nonumber.
 \end{align}

Now if we have an approximate solution $\widetilde{\vv}'$ satisfying $\|\widetilde{\vv}'-\widetilde{\vv}\|\le\eps/2\|\vv\|$, we check the three conditions respectively. For the first condition we have
\begin{align*}
	\|\PP_{(1+\gamma)\lambda}(\widetilde{\vv}'-\vv)\|
	& \le\|\PP_ {(1+\gamma)\lambda}(\widetilde{\vv}'-\widetilde{\vv})\|+\|\PP_{(1+\gamma)\lambda}(\widetilde{\vv}-\vv)\|\\
	& \le\|\widetilde{\vv}'-\widetilde{\vv}\|+\eps/2\|\vv\|\\
	& \le \eps/2\|\vv\|+\eps/2\|\vv\|\\
	& \le \eps\|\vv\|,
\end{align*}
while the third inequality uses the fact that $\PP_{(1+\gamma)\lambda}$ is a projection matrix.

For the second condition, we have
\begin{align*}
	\|(\I-\PP_{(1-\gamma)\lambda})\widetilde{\vv}'\| & \le \|(\I-\PP_{(1-\gamma)\lambda})(\widetilde{\vv}'-\widetilde{\vv})\|+\|(\I-\PP_{(1-\gamma)\lambda})\widetilde{\vv}\|\\
	& \le \|\widetilde{\vv}'-\widetilde{\vv}\|+\eps/2\|\vv\|\\
	& \le \eps/2\|\vv\|+\eps/2\|\vv\|\\
	& \le \eps\|\vv\|,
\end{align*}
where for the second inequality we use the fact that $\I-\PP_{(1-\gamma)\lambda}$ is also a projection matrix when $\PP_{(1-\gamma)\lambda}$ is a projection matrix.

For the last condition, we have 
\begin{align*}
 \|(\PP_{(1+\gamma)} - \PP_{(1-\gamma)\lambda}) (\widetilde{\vv}' - \vv)
\|
& \le
\|(\PP_{(1+\gamma)} - \PP_{(1-\gamma)\lambda}) (\widetilde{\vv}' - \widetilde{\vv})\|
+
\|(\PP_{(1+\gamma)} - \PP_{(1-\gamma)\lambda}) (\widetilde{\vv} - \vv)\|\\
	& \le \|\widetilde{\vv}'-\widetilde{\vv}\|+\|(\PP_{(1+\gamma)} - \PP_{(1-\gamma)\lambda}) \vv\|\\
	& \le \|(\PP_{(1+\gamma)} - \PP_{(1-\gamma)\lambda}) \vv\|+\eps\|\vv\|,
\end{align*}
where for the second inequality we use the fact that $\PP_{(1+\gamma)} - \PP_{(1-\gamma)\lambda}$ is a projection matrix.

Consequently, it suffices to have such $\widetilde{\vv}'$ that $\|\widetilde{\vv}'-\widetilde{\vv}\|\le\eps/2\|\vv\|$ when $\widetilde{\vv}=\frac{1}{2}(r_k^{\lambda\gamma}(\A^\top\A-\lambda\I)+\I)\vv$, note that
	$$
r^{\lambda\gamma}_k((\A^\top\A-\lambda\I))\vv= C (\A^\top\A-\lambda\I)\prod_{i=1}^k\frac{(\A^\top\A-\lambda\I)^2+c_{2i}\I}{(\A^\top\A-\lambda\I)^2+c_{2i-1}\I}\vv.
	$$
	
Suppose we have a procedure $\mathcal{C}_i(\vv),i\in[k]$ that can apply $\CC_i\vv$ for arbitrary $\vv$ to $\eps'$-multiplicative accuracy with probability $\ge1-\delta/k$, where here $$\CC_i=\frac{(\A^\top\A-\lambda\I)^2+c_{2i}\I}{(\A^\top\A-\lambda\I)^2+c_{2i-1}\I}.$$ Also we assume matrix vector product is accurate without loss of generality.\footnote{If in the finite-precision world, we assume arithmetic operations are carried out with $\Omega(\log(n/\eps))$ bits of precision, the result is still true by standard argument with a slightly different constant factor for the bounding coefficient.} Note that 
\begin{align*}
	\left\|\frac{(\A^\top\A-\lambda\I)^2+c_{2i}\I}{(\A^\top\A-\lambda\I)^2+c_{2i-1}\I}\right\| & \le M,\forall i\in[k],
\end{align*}
with $M=\beta_3k^4/\beta_2\gamma^2\lambda^2$. Here we use constants $\beta_2,\beta_3$ as stated in \cref{lem:bound}. Now we can use \cref{lem:stable} with the corresponding $M$ to show that 

Using a union bound, with probability $\ge1-\delta$ it holds that
$$
	\|\mathcal{C}_{k}(\mathcal{C}_{k-1}(\cdots\mathcal{C}_1(\vv)))-\prod\limits_{i=1}^{k}\CC_i\vv\|\le 2\eps' kM^{k-1}\|\vv\|,
$$
whenever $\eps'\le M/(2k)$.

Now we choose $$\tilde{\eps}_1=\min(\frac{M}{2k},\frac{\eps}{8kM^{k-1}}),\ \eps_1=\min(\frac{M}{2k},\frac{\eps}{8kM^{k-1}})/(\beta_3k^2)=\frac{\eps}{8\beta_3k^3M^{k-1}},$$ consider the following procedure as in \cref{alg:ISPCP},
\begin{align*}
\vv & \gets \RidgeSquare\left(\A,\lambda,c_{2i-1},(\A^\top\A-\lambda \I)^2\vv+c_{2i}\vv,\eps_1,\delta/k\right); \forall i\in[k].\\
\vv & \gets C(\A^\top\A-\lambda\I)\vv.
\end{align*}

The above choice of $\eps_1$ guarantees $\RidgeSquare\left(\A,\lambda,c_{2i-1},(\A^\top\A-\lambda \I)^2\vv+c_{2i}\vv,\eps_1,\delta/k\right)$ for all $i\in[k]$ can be abstracted as $\mathcal{C}_i(\vv)$ with $\tilde{\eps}_1$-accuracy and corresponding success probability. Using a union bound or successful events and also the fact that $\|C(\A^\top\A-\lambda\I)\|\le2$, we can argue that with probability $\ge1-\delta$, the output $\widetilde{\vv}'$ of $\ISPCP(\A,\vv,\lambda,\gamma,\eps,\delta)$ satisfy
$$
	\|\widetilde{\vv}'-\widetilde{\vv}\|\le 4\tilde{\eps}_1 kM^{k-1}\|\vv\|\le\eps/2\|\vv\|.
$$
As a result, the output of the algorithm satisfies conditions~(\ref{cond:PCP}) as desired.

\paragraph{Runtime:}
	
	The numerical constants $C, \{c_{i}\}_{i=1}^{2k}$ are precomputed. So the runtime will then be a total runtime of computing matrix vector products for $2k+1$ times, calling $k=O(\log(1/\eps)\log(1/\lambda\gamma))$ times $\RidgeSquare\left(\A,\lambda,c_{2i-1},(\A^\top\A-\lambda \I)^2\vv+c_{2i}\vv,\eps_1,\delta/k\right)$ for $i\in[k]$. We bound the two terms respectively.
		
	Computing matrix vector products takes time $O(k\nnz(\A))=\tilde{O}(\nnz(\A))$ since $k=\tilde{O}(1)$. %
	
	Using \cref{thm:square_solver_main}, since $\log(1/\eps_1)=O(\log(1/\eps)+\log k+k\log(M))=O(\log(1/\eps)+k)=\tilde{O}(1)$, the total runtime for solving squared systems is $\tilde{O}(k(\nnz(\A)+d\cdot\mathrm{sr}(\A)/(\gamma^2\lambda^2)))$. Further, it can be accelerated to $\tilde{O}(k\sqrt{\nnz(\A)d\cdot\mathrm{sr}(\A)}/(\lambda\gamma))$ when $\nnz(\A)\le d\cdot\mathrm{sr}(\A)/(\gamma^2\lambda^2)$.
	
	Combining these bounds  it gives a running time of \cref{alg:ISPCP} of
	$$
	\tilde{O}(\nnz(\A)+d\cdot\mathrm{sr}(\A)\frac{1}{\gamma^2\lambda^2}).
	$$
	When $\nnz(\A)\le d\cdot\mathrm{sr}(\A)/(\gamma^2\lambda^2)$, it can be accelerated to $\tilde{O}(\sqrt{\nnz(\A)d\cdot\mathrm{sr}(\A)}/(\lambda\gamma))$.
	\end{proof}
	
	Since we assume $\lambda_1\in[1/2,1]$ here, $\kappa=1/\lambda$. We can write it as $$\tilde{O}\left(\nnz(\A)+\sqrt{\nnz(\A)\cdot d\cdot\mathrm{sr}(\A)}\kappa/\gamma\right)$$ by noticing the preprocessing for $\A$ is just setting $\lambda\leftarrow \Theta(\lambda/\lambda_1)$.

\subsection{PCR Solver}
\label{ssec:PCR}

Previous work as shown that solving PCR can be reduced to solving PCP together with ridge regression solver. This reduction was first proposed in~\citet{FMMS16} and used in subsequent work~\cite{AL17}. The idea is to first compute $\vv^*=\PP_\lambda(\A^\top \bb)$ using PCP and then apply $(\A^\top\A)^\dagger\vv^*$ stably by some polynomial approximation. More specifically, it is achieved through the following procedure.
\begin{equation}
\label{red:PCPtoPCR}
\begin{aligned}
	\sss_0 & \leftarrow\mathcal{A}_\PCP(\A^\top\bb)\\
	\sss_1 & \leftarrow\RidgeReg(\A,\lambda,\sss_0),\ \forall m=1,2,\cdots,k-1;\\
	\sss_{m+1} & \leftarrow\sss_1+\lambda\cdot \RidgeReg(\A,\lambda,\sss_m). 
\end{aligned}
\end{equation}

Here $\RidgeReg$ is a Ridge Regression Solver defined in \cref{def:rigreg} and $\ISPCP$ is the $\eps$-approximate PCP algorithm  specified in \cref{alg:ISPCP}. Such a reduction enjoys the following guarantee:

\begin{lemma}[Reduction: from PCR to PCP]\label{red:PCP}
	For fixed $\lambda$, $\eps\in(0,1)$ and $\gamma\in(0,2/3]$, and $\A$ with singular values no more than 1, given an $O(\eps/k^2)$-approximate ridge regression solver $\RidgeReg(\A,\lambda,\cdot)$ and an $O(\gamma\epsilon/k^2)$-approximate PCP solver $\mathcal{A}_\PCP(\cdot)$. Running the procedure for $k=\Theta(\log(1/\eps\gamma))$ iterations and  outputting the final $\sss_{k}$ gives an $\eps$-approximate PCR algorithm.
\end{lemma}

For completeness, we give the following \cref{alg:ISPCR} for $\eps$-approximate PCP solver and give the proof for its theoretical guarantee as stated in~\cref{thm:pcr_main}.

\begin{algorithm}[H]
\DontPrintSemicolon
  \KwInput{$\A\in\R^{n\times d}$ properly rescaled, $\bb\in\R^{d}$ regressand, $\lambda>0$ eigenvalue threshold, $\gamma\in(0,2/3]$ unitless eigengap, $\eps$ desired accuracy, $\delta$ probability parameter.}
  \KwOutput{A vector $\xx\in\R^{n}$ that solves PCR $\eps$-approximately.}
  Set $k\gets\Theta(\log(1/\eps\gamma))$, $\eps_1\gets O(\frac{\gamma\eps}{k^2})$,$\eps_2\gets O(\frac{\eps}{k^2})$\;
  $\xx \gets \ISPCP(\A,\A^\top\bb,\lambda,\gamma,\eps_1,\delta/4)$\;
  $\xx_0\gets \RidgeReg(\A,\lambda,\xx,\eps_2,\delta/4)$\;
  \For{$i\gets 1$ to $k-1$}
  {
  Compute $\xx=\lambda\cdot \RidgeReg(\A,\lambda,\xx,\eps_2,\delta/2(k-1))+\xx_0$\;
  }
  \caption{$\ISPCR(\A,\bb,\lambda,\gamma,\eps,\delta)$}
  \label{alg:ISPCR}
\end{algorithm}	
	
\begin{proof}[Proof of \cref{thm:pcr_main}]\

\paragraph{Approximation:} 
	It follows directly from \cref{red:PCP}.
\paragraph{Runtime:} 
	
	The total running time consists of one call of $\ISPCP$ and $k=\Theta(\log(1/\eps\gamma))$ calls of $\RidgeReg$, with particular parameters specified in \cref{alg:ISPCR} this leads to a runtime 
	$$\tilde{O}(\nnz(\A)+d\cdot\mathrm{sr}(\A)\frac{1}{\gamma^2\lambda^2})+k\cdot\tilde{O}(\nnz(\A)+d\cdot\mathrm{sr}(\A)\frac{1}{\lambda})=\tilde{O}(\nnz(\A)+d\cdot\mathrm{sr}(\A)\frac{1}{\gamma^2\lambda^2})),
	$$
	and an accelerated runtime 
	$$
	\tilde{O}(\frac{1}{\gamma\lambda}\sqrt{\nnz(\A)d\cdot\mathrm{sr}(\A)})
	$$
	when $\nnz(\A)\le d\cdot\mathrm{sr}(\A)/(\gamma\lambda)^2$, which proves the result by noticing $\kappa=1/\lambda$ for rescaled $\A$.
	
\end{proof}

\section{Numerical Experiments}
\label{sec:experiment}

We evaluate our proposed algorithms following the settings in~\citet{FMMS16,AL17}. As the runtimes in \cref{thm:pcp_main,thm:pcr_main} show improvement compared with ones in previous work~\cite{FMMS16,AL17} when $\nnz(\A)/\gamma\gg d^2\kappa^2/\gamma^2$, we  pick data matrix $\A$ such that $\kappa=\Theta(1)$ and $n\gg\frac{d}{\gamma}$ to corroborate the theoretical results.

Since experiments in several papers~\cite{FMMS16,AL17} have studied
the reduction from PCR to PCP (see \cref{red:PCP}), here we only show results regarding solving PCP problems. In all figures below, the $y$-axis denote the relative error measured in $\|\mathcal{A}_{\PCP}(\vv)-\PP_\lambda\vv\|/\|\PP_\lambda\vv\|$ and $x$-axis denote the total number of vector-vector products to achieve corresponding accuracy.\\

\paragraph{Datasets.} Similar to that in previous work~\cite{FMMS16,AL17}, we set $\lambda=0.5, n=2000,d=50$ and form a matrix $\A=\UU\LLambda^{1/2}\V\in\R^{2000\times 50}$. Here, $\UU$ and $\V$ are random orthonormal matrices, and $\SSigma$ contains randomly chosen singular values $\sigma_i=\sqrt{\lambda_i}$. Referring to $[0,\lambda(1-\gamma)]\cup[\lambda(1+\gamma),1]$ as the \emph{away-from-$\lambda$ region}, and $\lambda(1-\gamma)\cdot[0.9,1]\cup\lambda(1+\gamma)\cdot[1,1.1]$ as the \emph{close-to-$\lambda$ region}, we generate $\lambda_i$ differently to simulate the following three different cases:

\begin{enumerate}[label=\roman*]
  \item Eigengap-Uniform Case: generate all $\lambda_i$ uniformly in the away-from-$\lambda$ region.
  \item Eigengap-Skewed Case: generate half the $\lambda_i$ uniformly in the away-from-$\lambda$ and half uniformly in the close-to-$\lambda$ regions.
  \item No-Eigengap-Skewed Case: uniformly generate half in $[0,1]$, and half in the close-to-$\lambda$ region.
\end{enumerate}

\paragraph{Algorithms.} We implemented the following algorithms and compared them in the above settings:%
\begin{enumerate}
\item \emph{polynomial}: the $\PCProj$ algorithm in~\citet{FMMS16}.
\item \emph{chebyshev}: the $\QuickPCP$ algorithm in~\citet{AL17}.
\item \emph{lanczos}: the algorithm using Lanczos method discussed in Section 8.1 of~\citet{musco2018stability}.
\item \emph{rational}: the $\ISPCP$ algorithm (see~\cref{alg:ISPCP}) proposed in our paper.
\item \emph{rlanczos}: the algorithm using rational Lanczos method~\cite{gallivan1996rational} combined with $\ISPCP$. (See~\cref{App:rlanczos} for a more detailed discussion.)
\item \emph{slanczos}: the algorithm using Lanczos method~\cite{musco2018stability} with changed search space from form $f(\frac{x-\lambda}{x+\lambda})$ into $f(\frac{(x-\lambda)(x+\lambda)}{(x-\lambda)^2+\gamma(x+\lambda)^2})$ for approximation ($f$ polynomial, $x\gets\A^\top\A$). (See~\cref{App:slanczos} for a more detailed discussion.) %
\end{enumerate}

We remark that 1-3 are algorithms in previous work; 4 is an exact implementation of $\ISPCP$ proposed in the paper; 5,\ 6 are variants of $\ISPCP$ combined with Lanczos method, both using the squared system solver. Algorithms 5,\ 6 are explained in greater detail in \cref{App:exp}.

\begin{figure}[htb]%
\centering
\subfigure[Eigengap-Uniform Case]{%
\includegraphics[width=0.33\textwidth]%
       {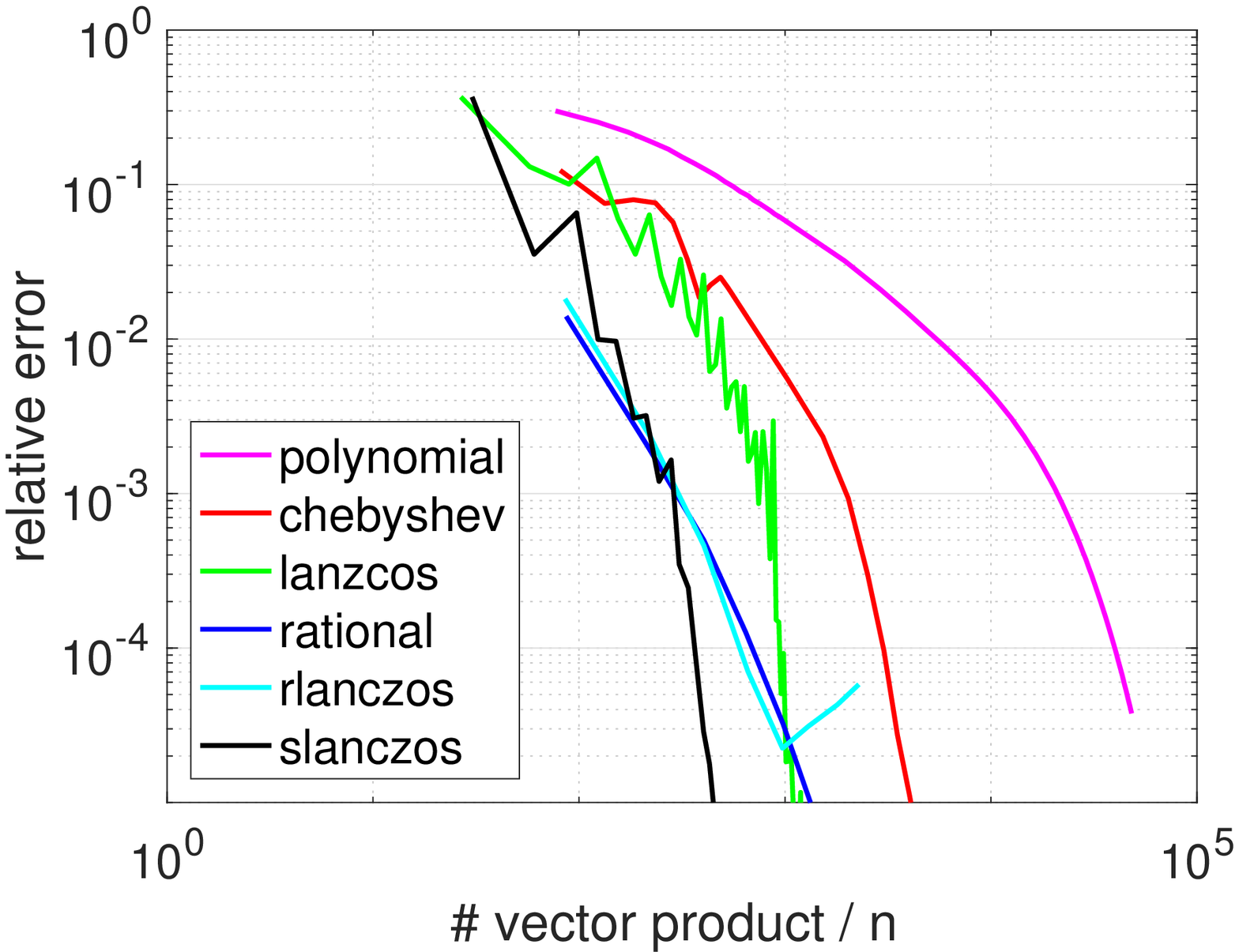}
}%
\subfigure[Eigengap-Skewed Case]{%
\includegraphics[width=0.33\textwidth]%
       {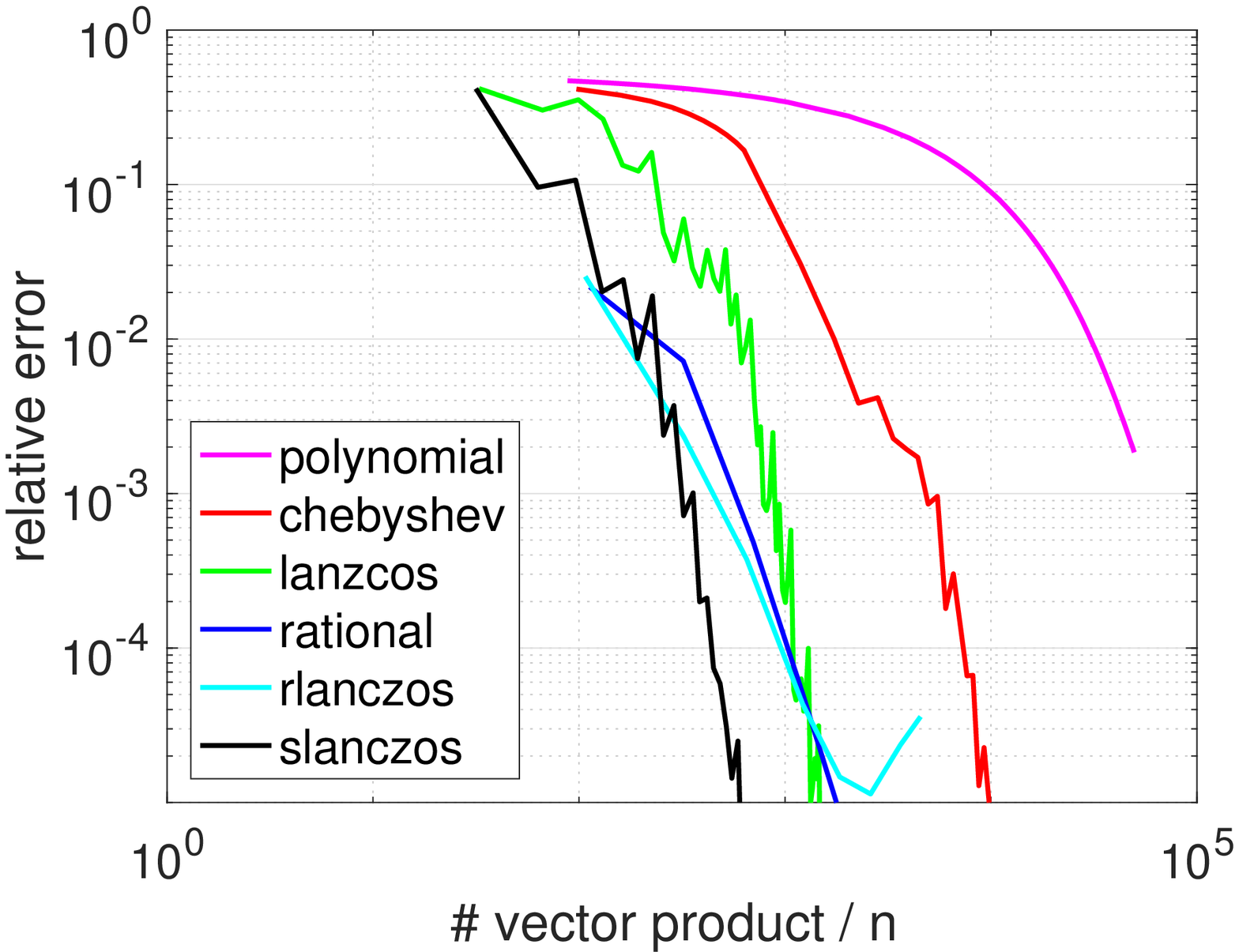}
}%
\subfigure[No-Eigengap-Skewed Case]{%
\includegraphics[width=0.33\textwidth]%
       {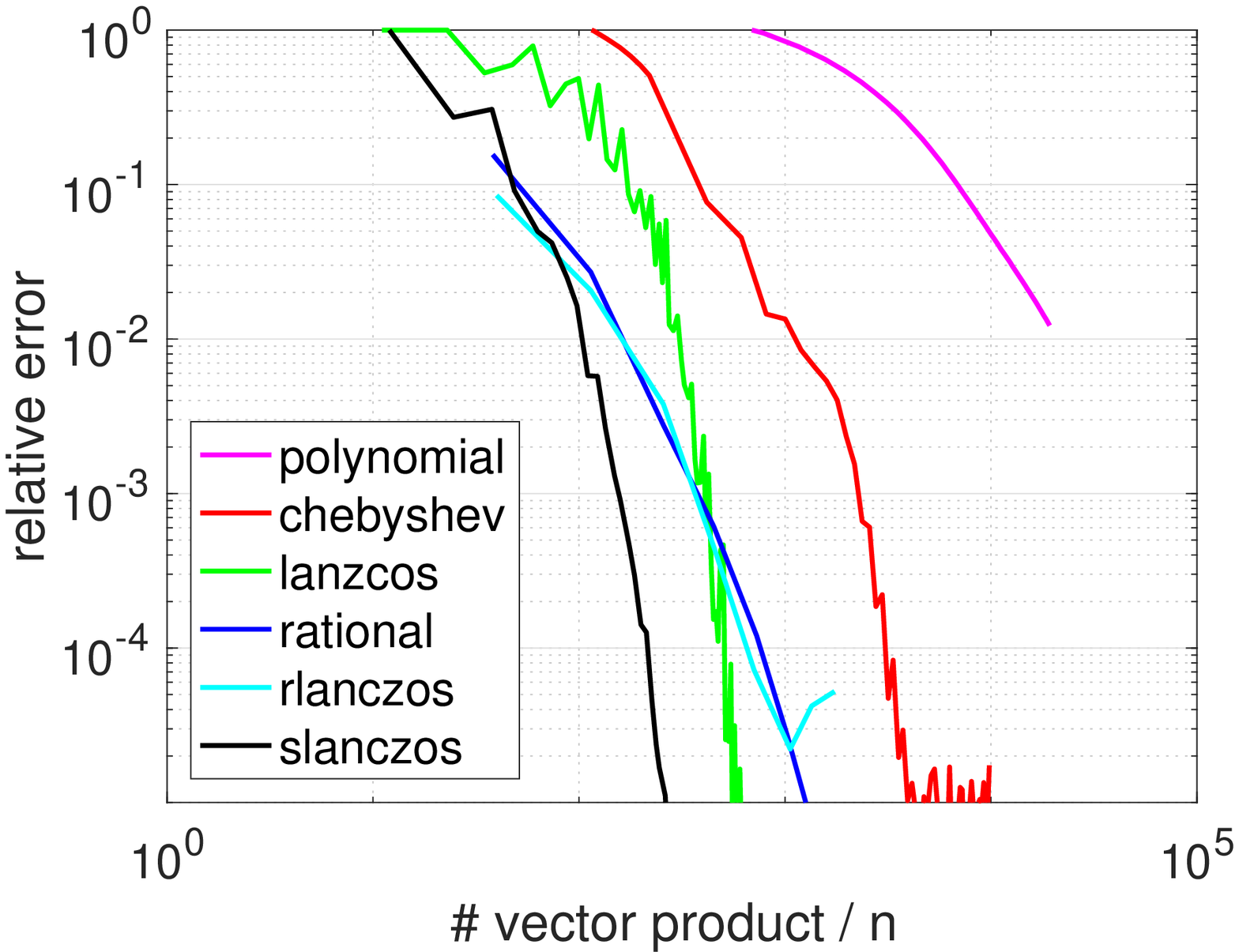}
}\\%
\caption{Synthetic Data: $n=2000$, $d=50$, $\lambda=0.5$, $\gamma=0.05$. }
\label{fig:gap}
\vspace{-0.5cm}
\end{figure}

\begin{figure}[htb]%
\centering
\subfigure[n=5000]{%
\includegraphics[width=0.33\textwidth]%
	{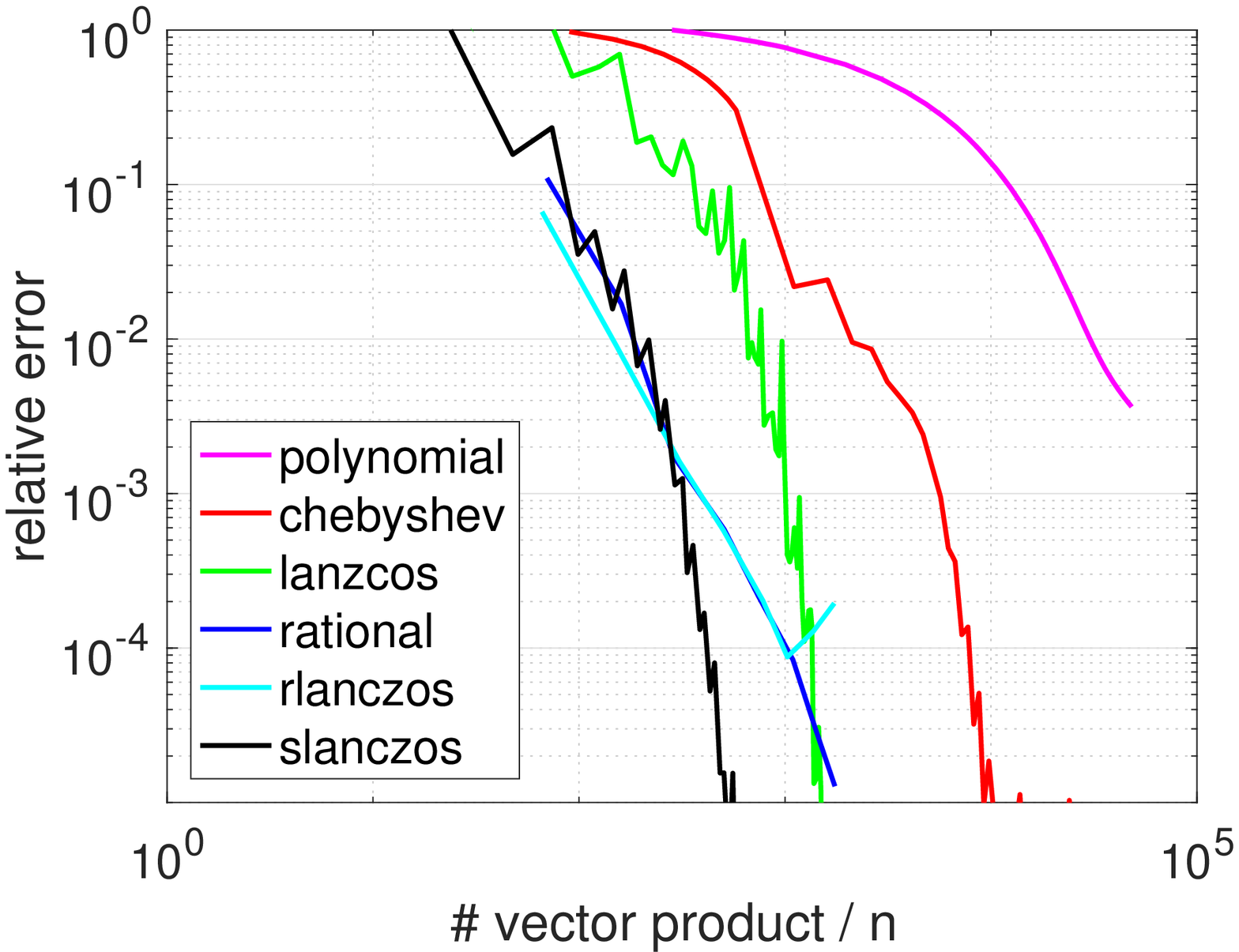}}
\subfigure[n=10000]{%
\includegraphics[width=0.33\textwidth]%
	{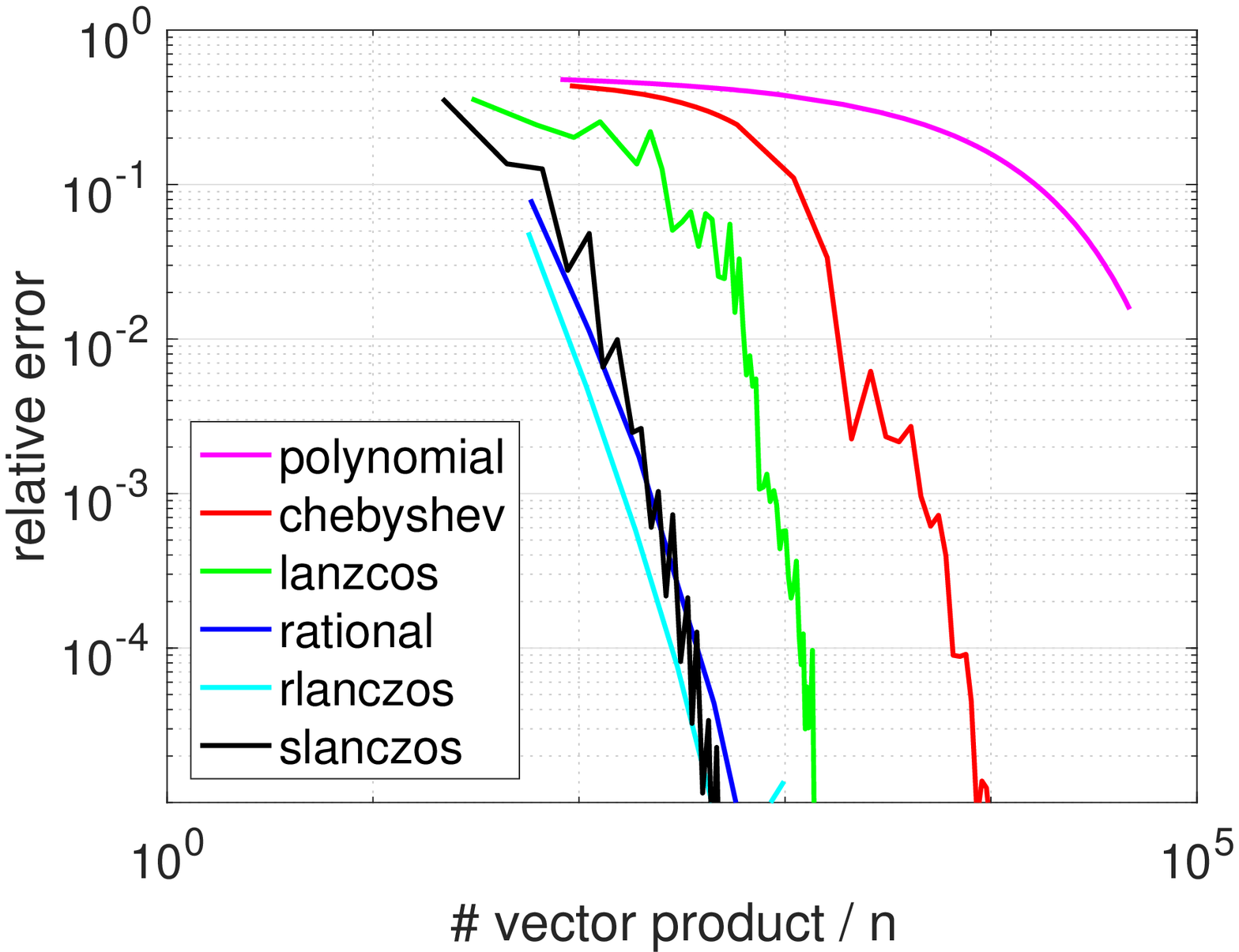}}%
\subfigure[n=20000]{%
\includegraphics[width=0.33\textwidth]%
	{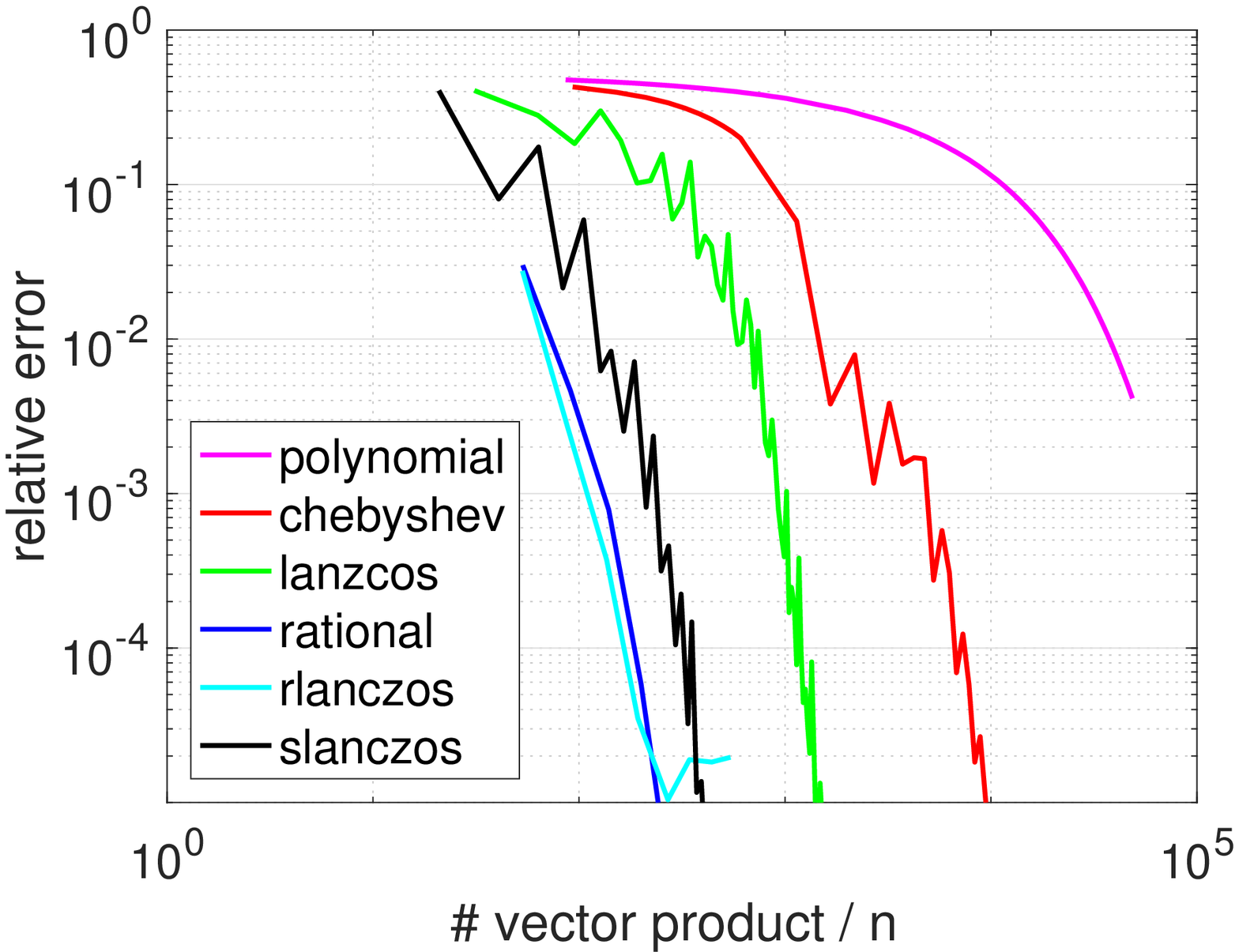}}\\%
\caption{Synthetic Data: Changing $n$, $d=50$, $\lambda=0.5$, $\gamma=0.05$. No-Eigengap-Skewed Case.}
\label{fig:n}
\end{figure}

There are several observations from the experiments:\\
\begin{itemize}
\item For different eigenvalue distributions, (4-6) in general outperform all existing methods (1-3) in most accuracy regime in terms of number of vector products as shown in \cref{fig:gap}.\\
\item In no-eigengap case, all methods get affected in precision. This is due to the projection error of eigenvalues very close to eigengap, which simple don't exist in Eigengap cases. Nevertheless, (6) is still the most accurate one with least computation cost, as shown in \cref{fig:gap}.\\
\item When $n$ gets larger, (4,5) tends to enjoy similar performance, outperforming all other methods including (6), as shown in \cref{fig:n}. This aligns with theory that runtime of (4,5) is dominated by $\nnz(\A)$ while runtime of (6) is dominated by $\nnz(\A)/\sqrt{\gamma}$ (see \cref{thm:slanczos} for theoretical analysis of \emph{slanczos}), demonstrating the power of nearly-linear runtime of $\ISPCP$ proposed.
\end{itemize}

\section{Conclusion}
\label{sec:conclusion}

In this paper we provided a new linear algebraic primitive, asymmetric SVRG for solving squared ridge regression problems, and showed that it lead to nearly-linear-time algorithms for PCP and PCR. Beyond the direct running time improvements, this work shows that running time improvements can be achieved for fundamental linear-algebraic problems by leveraging stronger subroutines than standard ridge regression. The improvements we obtain for PCP, demonstrated theoretically and empirically, we hope are just the first instance of a more general approach for improving the running time for solving large-scale machine learning problems.

\medskip

{\small
\bibliographystyle{abbrvnat}
\bibliography{references.bib}
}

\newpage
\appendix

\section*{Appendix}

\section{Ridge Regression Solver $\RidgeReg$}
\label{App:solver}

In the reduction from PCP to PCR stated in \cref{sec:PCP}, a blackbox solver $\RidgeReg(\A,\lambda,\sss)$ is needed as the case of~\citet{FMMS16,AL17}. We formally define it as follows.

\begin{definition}[Ridge Regression Solver]\label{def:rigreg}
Given $\sss\in\R^d$, and exact solution $\xx^*=(\A^\top \A+\mu\I)^{-1}\sss$, an algorithm %
$\RidgeReg(\A,\mu,\sss)$ is an $\eps$-approximate ridge regression blackbox solver if for all
$\epsilon>0$, it returns a solution $\tilde{\xx}$ satisfying $ \|\tilde{\xx}-\xx^*\|\le\eps\|\sss\|.$
\end{definition}
\begin{theorem}[Blackbox Solver for Ridge Regression]
There is a blackbox solver for ridge regression with runtime $\tilde{O}(\nnz(\A)+\sqrt{\nnz(\A)d\cdot\mathrm{sr}(\A)\kappa_\mu})$.	
\end{theorem}

An algorithm achieving the aboveboard  theoretical guarantee can be found in~\citet{SZ14,FGKS15,LMH15}.

We remark that when depending on structural properties of $\A$ (denoting the maximum row sparsity of $\A$ as $s(\A)$ such that $s(\A)\le d$), the running time for ridge regression can be further improved into $\tilde{O}((n+\sqrt{n\cdot\sr(\A)\kappa_\mu})\cdot s(\A))$ through direct acceleration Katyusha~\cite{allen2017katyusha} or accelerated coordinate descent methods \cite{nesterov2017efficiency}. Using a leverage-score sampling idea~\cite{agarwal2017leverage}, one can obtain an even more fine-grained running time $\tilde{O}((n+\sqrt{d\cdot\sr(\A)\kappa_\mu})\cdot s(\A))$ that improves the previous in most settings. Because these algorithms have running time highly depend on structure of $\A$, we didn't use them for computing runtime when treating ridge regression solver as a blackbox.

\section{Proofs for Results in \cref{sec:PCP}}
\label{App:PCP}

\subsection{Properties of Zolotarev Rational}
\label{App:prop}

\begin{proof}[Proof of \cref{lem:bound}]

(1) The elliptic integral has taylor series~\cite{Calson10} as follows: 
\[ 
K(\mu)\defeq\int_0^1\frac{dt}{\sqrt{(1-t^2)(1-\mu^2t^2)}}=\frac{\pi}{2}\sum\limits_{n=0}^{\infty}\left(\frac{(2n)!}{2^{2n}(n!)^2}\right)^2\mu^{2n}
\]
Using the Stirling formula $n!\approx\sqrt{2\pi n}(n/e)^n$, $\exists$ constants $C_1,C_2<\infty$ such that, %
\begin{align*}
	K(\mu) & = \frac{\pi}{2}\sum\limits_{n=0}^\infty\left(\frac{(2n)!}{2^{2n}(n!)^2}\right)^2\mu^{2n}
	 \le C_1\sum\limits_{n=1}^\infty\frac{\mu^{2n}}{n}\\
	& \le C_2\int_1^\infty\frac{\mu^{2t}}{t}dt
	\stackrel{(i)}{\le} C_2\int_{\sqrt{\gamma}}^\infty \frac{e^{-t}}{t}dt\\
	& \stackrel{(ii)}{=} C_2E_1(\sqrt{\gamma})
\end{align*}
where we use $(i)$ $\mu^2\le(1-\sqrt{\gamma})^2\le1-\sqrt{\gamma}\le\exp(-\sqrt{\gamma})$, and $(ii)$ change of variable $t\leftarrow \sqrt{\gamma}t$ and definition of exponential integral that $E_1(z)\defeq\int_z^\infty \frac{e^{-t}}{t}dt$

By the convergence series of exponential integral~\cite{abramowitz1965handbook}, this can be written for $z>0$ as %
\[
E_1(z)=C_3-\log(z)-\sum\limits_{k=1}^\infty \frac{(-z)^k}{k\cdot k!}
\]
where $C_3$ finite is the Euler-Mascheroni constant. Using this, we have
\begin{align*}
	K(\mu) & \le C_2 E_1(\sqrt{\gamma})\lesssim \log(1/\gamma) ,\text{ as }\gamma\rightarrow 0,
\end{align*} 
where $\lesssim$ is hiding constant $C$ multiplicatively and this yields (1).

(2) By definition of elliptic integral we have,
\begin{align*}
	K'  \defeq \int_0^{\pi/2}\frac{d\theta}{\sqrt{1-\gamma'^2\mathrm{sin}^2\theta}}
	\enspace \text{ and } \enspace
	c_i \defeq \gamma^2\frac{\mathrm{sn}^2(\frac{i K'}{2k+1};\gamma')}{\mathrm{cn}^2(\frac{i K'}{2k+1};\gamma')}
	\enspace
	\text{ for all }
	i\in[2k].
\end{align*}

Notice by equivalent definition for each $i\in[2k]$ we have, $$\frac{\mathrm{sn}(\frac{i K'}{2k+1};\gamma')}{\mathrm{cn}(\frac{i K'}{2k+1};\gamma')}=\tan\phi\text{ where }\frac{iK'}{2k+1}=\int_0^\phi\frac{d\theta}{\sqrt{1-\gamma'^2\mathrm{sin}^2\theta}}.$$
Consequently, we know $c_i$ is monotonously decreasing in $i$ as $\phi$ itself is monotonously decreasing.

Also, since $\gamma'^2 = 1-\gamma^2$, we have $\sqrt{1-\gamma'^2\mathrm{sin}^2\theta}=\sqrt{\cos^2\theta+\gamma^2\sin^2\theta}$ which $\rightarrow 1$ when $\theta\rightarrow0$ and $\rightarrow \gamma$ when $\theta\rightarrow\pi/2$. From that we know 
\begin{align*}
\frac{1}{2k+1}\int_0^{\pi/2}\frac{d\theta}{\sqrt{1-\gamma'^2\mathrm{sin}^2\theta}} \le\int_0^\phi\frac{d\theta}{\sqrt{1-\gamma'^2\mathrm{sin}^2\theta}} \le\frac{2k}{2k+1}\int_0^{\pi/2}\frac{d\theta}{\sqrt{1-\gamma'^2\mathrm{sin}^2\theta}} ,\\
\end{align*}
and thus
\begin{align*}
\frac{1}{2k+1}\cdot\frac{\pi}{2} \le\phi\le\frac{\pi}{2}-\frac{\gamma K'}{2k+1},~\text{ where by definition }\gamma K'\in(0,\pi/2).
\end{align*}
Using monotonicity for $\tan(\phi)$ on $\phi\in[0,\pi/2]$, $\tan(\phi)\ge\phi$ and $\sin(\phi)\ge\phi/2$ for all $\phi\in(0,\pi/2)$, we know $\exists \beta_2,\beta_3 \in (0,\infty)$ such that %
\begin{align*}
\frac{\mathrm{sn}^2(\frac{K'}{2k+1};\gamma')}{\mathrm{cn}^2(\frac{K'}{2k+1};\gamma')} &\ge \beta'_2/(2k+1)^2\ge\beta_2/k^2\\
\frac{\mathrm{sn}^2(\frac{2k K'}{2k+1};\gamma')}{\mathrm{cn}^2(\frac{2k K'}{2k+1};\gamma')} & \le \beta_3'\left(\frac{2k+1}{\gamma K'}\right)^2\le \frac{\beta_3k^2}{\gamma^2}.
\end{align*}
By definition of $K'$ we know $K'\ge\Theta(1)$, yielding that $c_1\ge\beta_2\gamma^2/k^2$ and $c_{2k}\le\beta_3k^2$.
\end{proof}

\section{Proofs for Results in \cref{sec:SVRG}}
\label{App:SVRG}

\subsection{Direct Methods Runtime}
\label{App:direct}
We first state the results and corresponding proofs of using direct methods stated in the beginning of \cref{sec:SVRG}. Consider the squared system  $\left((\A^\top\A-c\I)^2+\mu^2\I\right)\xx=\vv$ for given $\A\in\R^{n\times d},\vv\in\R^d$, $\mu>0$ and $c\in[0,\lambda_1]$.

The following theorem gives a bound of the running time of solving this system using accelerated gradient descent~\cite{bubeck2015geometric,nesterov2018lectures}.

\begin{theorem}[Direct AGD Runtime]
\label{thm:direct-AGD}
Consider iteration $$\xx_{t+1}=\xx_t-\gamma\left[\bigl((\A^\top\A-c\I)^2+\mu^2\I\bigr)\xx_t-\vv\right]+\beta(\xx_t-\xx_{t-1}).$$ Choosing $\gamma=\frac{4}{\lambda_1+2\mu},\beta=\frac{\lambda_1}{\lambda_1+2\mu}$ under above condition, 
the total running time to get an $\eps$-approximate solution is $\tilde{O}(\nnz(\A)\lambda_1/\mu)$.
\end{theorem}

Now we turn to analyzing SVRG directly applied to the squared system as follows %
\begin{equation}
\left((\A^\top \A-c \I)^2+\mu^2\I\right)\xx=\vv.
	\label{cond:quad}
\end{equation}

We take the step 
\begin{equation}
\label{eqn:SVRG-one-step-quad}
\begin{aligned}
\xx_{t+1} & = \xx_t-\frac{\eta}{p_{ij}}\bigl(\M_{ij} \xx_t-\M_{ij}\xx_0+p_{ij}(\M\xx_0-\vv)\bigr)\\
\text{ where }
\M_{ij} & \defeq \aaa_i\aaa_i^\top \aaa_j\aaa_j^\top -2c\frac{\|\aaa_j\|^2}{\|\A\|_\mathrm{F}^2}\aaa_i\aaa_i^\top +(c^2+\mu^2)\frac{\|\aaa_j\|^2\|\aaa_i\|^2 }{\|\A\|_\mathrm{F}^4}\I\\
\text{ and } p_{ij} & \propto\norm{\aaa_i}^2\norm{\aaa_j}^2,\ \forall i,j\in[n].
\end{aligned}
\end{equation}
Such update gives the following variance bound.

\begin{lemma}[Variance bound for solving squared system directly]
\label{lem:var-quad}
For problem~\eqref{cond:quad}, the variance incurred in~\eqref{eqn:SVRG-one-step-quad} is bounded by $S=O((\lambda_1^2\|\A\|_\mathrm{F}^4+(\lambda_1^2+\mu^2)^2)/\mu^2)$, i.e. there exists $0<C<\infty$ that
$$
\sum\limits_{i,j\in[n]}\frac{1}{p_{ij}}\|\M_{ij}\xx_t-\M_{ij}\xx_0 +p_{ij} (\M\xx_0-\vv)\|^2\le \frac{C(\lambda_1^2\|\A\|_\mathrm{F}^4+(\lambda_1^2+\mu^2)^2)}{\mu^2}[\norm{\xx_t-\xx^*}^2+\norm{\xx_0-\xx^*}^2], $$
where $C<\infty$ is a numerical constant and $\|A\|_{\mathrm{F}}^2\le d^2$.
\end{lemma}
\begin{proof}
Notice that $\nabla\psi_{ij} = [\aaa_i\aaa_i^\top \aaa_j\aaa_j^\top -2c\frac{\|\aaa_j\|^2}{\|\A\|_\mathrm{F}^2}\aaa_i\aaa_i^\top +(c^2+\mu^2)\frac{\|\aaa_j\|^2\|\aaa_i\|^2 }{\|\A\|_\mathrm{F}^4}\I]\xx-\frac{\|\aaa_j\|^2\|\aaa_i\|^2}{\|\A\|_\mathrm{F}^4} \vv$, by bounding directly and summing up all terms $i,j\in[n]$, we get  
\begin{align*}
    & \sum\limits_{i,j\in[n]}\frac{1}{p_{ij}}\|\M_{ij}\xx-\M_{ij}\xx^*\|^2\\
    & =(\xx-\xx^*)^\top\left(\sum_{i,j\in[n]}\frac{1}{p_{ij}}\M_{ij}^\top\M_{ij}\right)(\xx-\xx^*)\\
    = &\sum\limits_{i,j\in[n]}\frac{\|\A\|_\mathrm{F}^4}{\|\aaa_i\|^2\|\aaa_j\|^2}\|\biggl(\aaa_i\aaa_i^\top \aaa_j\aaa_j^\top -2c\frac{\|\aaa_j\|^2}{\|\A\|^2_\mathrm{F}}\aaa_i\aaa_i^\top +\frac{\|\aaa_i\|^2\|\aaa_j\|^2}{\|\A\|_\mathrm{F}^4}(c^2+\mu^2)I\biggr)(\xx-\xx^*)\|^2\\
    = & (\xx-\xx^*)^\top \biggl(\sum\limits_{i,j\in[n]}\bigl(\frac{\|\A\|_\mathrm{F}^4}{\|\aaa_j\|^2}\aaa_j\aaa_j^\top \aaa_i\aaa_i^\top \aaa_j\aaa_j^\top +4c^2\|\aaa_j\|^2\aaa_i\aaa_i^\top +(c^2+\mu^2)^2\frac{\|\aaa_i\|^2\|\aaa_j\|^2}{\|\A\|_\mathrm{F}^4}\I\\
    & -2c\|\A\|_\mathrm{F}^2(\aaa_j\aaa_j^\top \aaa_i\aaa_i^\top +\aaa_i\aaa_i^\top \aaa_j\aaa_j^\top )+(c^2+\mu^2)(\aaa_j\aaa_j^\top \aaa_i\aaa_i^\top +\aaa_i\aaa_i^\top \aaa_j\aaa_j^\top )\\
    & -4c(c^2+\mu^2)\frac{\|\aaa_j\|^2}{\|\A\|^2_{\mathrm{F}}}\aaa_i\aaa_i^\top \bigr)\biggr)(\xx-\xx^*)\\
    = & (\xx-\xx^*)^\top \bigl(\|\A\|_\mathrm{F}^4(\A^\top \A)(\A^\top \A)+4c^2\|\A\|_\mathrm{F}^2(\A^\top \A)+(c^2+\mu^2)^2\I\\
    & - 4c^2(\A^\top \A)(\A^\top \A)+2(c^2+\mu^2)(\A^\top \A)(\A^\top \A)-4c(c^2+\mu^2)(\A^\top \A)\bigr)(\xx-\xx^*)\\
    \stackrel{(i)}{\le} & (\xx-\xx^*)^\top \left(\|\A\|_\mathrm{F}^4\lambda_1^2\I+4c^2\lambda_1\|\A\|_\mathrm{F}^2\I+(c^2+\mu^2)^2\I + 2(c^2+\mu^2)\lambda_1^2\right)(\xx-\xx^*)\\
    \le & C\left(\lambda_1^2\|\A\|_\mathrm{F}^4+4c^2\lambda_1\|\A\|_\mathrm{F}^2+(c^2+\mu^2)(c^2+\mu^2+2\lambda_1^2)\right)\|\xx-\xx^*\|^2\\
    \stackrel{(ii)}{\le} & C\left(\lambda_1^2\|\A\|_\mathrm{F}^4+(\lambda_1^2+\mu^2)^2\right)\|\xx-\xx^*\|^2,
\end{align*}
where we use $\A^\top\A\preceq\lambda_1\I$ and $c\in[0,\lambda_1]$ in (i) by condition, and (ii)
holds for some constant $0<C<\infty$. Then we conclude from the fact that $\|\xx+\xx'\|^2\le2\|\xx\|^2+2\|\xx'\|^2$ and $\mathbb{E}\|\xx-\mathbb{E}\xx\|^2\le\mathbb{E}\|\xx\|^2$  that
$$
\sum\limits_{i,j\in[n]}\frac{1}{p_{ij}}\|\M_{ij}\xx_t-\M_{ij}\xx_0 +p_{ij} (\M\xx_0-\vv)\|^2\le \frac{C(\lambda_1^2\|\A\|_\mathrm{F}^4+(\lambda_1^2+\mu^2)^2)}{\mu^2}[\norm{\xx_t-\xx^*}^2+\norm{\xx_0-\xx^*}^2].
$$
\end{proof}

\begin{theorem}[Direct SVRG Runtime]
\label{thm:direct-SVRG}
For problem~\eqref{cond:quad}, the SVRG algorithm applied with~\eqref{eqn:SVRG-one-step-quad} returns with probability $\ge1-\delta$ an $\eps$-approximate solution in $\tilde{O}(\nnz(\A)+d\cdot \mathrm{sr}(\A)^2\lambda_1^4/\mu^4)$ time. An accelerated variant of it improves the runtime to $\tilde{O}(\nnz(\A)+\nnz(\A)^{3/4}d^{1/4}\sr(\A)^{1/2}\lambda_1/\mu)$.
\end{theorem}

The proof is a direct combination of \cref{lem:gen-per-iter} with the proving technique for \cref{thm:main-asySVRG}. We omit it here as the procedure and argument are basically the same.

From that we can get a direct SVRG solver of squared system $\left((\A^\top\A-c\I)^2+\mu^2\I\right)\xx=\vv$ that outputs $\eps$-approximate solution with high probability with running time $\tilde{O}(\nnz(\A)+d\cdot \mathrm{sr}(\A)^2\lambda_1^4/\mu^4)$.

Because of the high complexity of above methods (either not nearly-linear in AGD or squaring problem dimension, i.e. having $\mathrm{sr}(\A)^2\lambda_1^4/\mu^4$ as condition number in SVRG), a new insight is required to better solve such systems. The technique we develop for this purpose is to 'decouple' the squared matrix at the cost of asymmetry, formally introduced by the following reduction.

\section{Proofs for Results in \cref{App:main}}
\label{App:main-sup}

\begin{proof}[Proof of \cref{lem:stable}]

By induction, for $k=1$, this is true since $\|\mathcal{C}_1(\vv)-\CC_1\vv\|\le\eps\|\vv\|$. Suppose this is true for $i$, i.e. $\|\mathcal{C}_{i}(\mathcal{C}_{i-1}(\cdots\mathcal{C}_1(\vv)))-\prod\limits_{j=1}^i\CC_j\vv\|\le 2\eps iM^{i-1}\|\vv\|$, then for $i+1$, we have
\begin{align*}
	\|\mathcal{C}_{i+1}(\mathcal{C}_{i}(\cdots\mathcal{C}_1(\vv)))-\prod\limits_{j=1}^{i+1}\CC_j\vv\| 
	\le & \|\mathcal{C}_{i+1}(\mathcal{C}_{i}(\cdots\mathcal{C}_1(\vv)))-\CC_{i+1}\mathcal{C}_{i}(\cdots\mathcal{C}_1(\vv))\|\\
	& + \|\CC_{i+1}(\mathcal{C}_{i}(\cdots\mathcal{C}_1(\vv))-\prod\limits_{j=1}^i\CC_j\vv)\|\\
	\le & \eps\|\mathcal{C}_{i}(\cdots\mathcal{C}_1(\vv))\|+\|\CC_{i+1}\|\|\mathcal{C}_{i}(\cdots\mathcal{C}_1(\vv))-\prod\limits_{j=1}^i\CC_j\vv\|\\
	\le & \eps(M^i\|\vv\|+2\eps i M^{i-1}\|\vv\|)+2\eps iM^i\|\vv\|\\
	\le & \eps M^i(2i+1+\frac{2\eps i}{M})\|\vv\|\\
	\le & 2\eps(i+1)M^i\|\vv\|,
\end{align*}
where the last inequality uses the condition $2\eps i\le2\eps k\le M$.
\end{proof}

\subsection{Square-root Computation}

Here we prove \cref{thm:square_root_solver_main}. The approach is very similar for developing the PCP solver as in~\cref{ssec:PCP_main}. We first introduce a rational function that when applied to any given PSD matrix $\M$ well approximates the square-root of itself for matrix $\mu\I\preceq\M\preceq\I$. Note this will immediately generalize to the case when $\mu\I\preceq\M\preceq\lambda\I$ by first getting a constant approximation $\tilde{\lambda}$ of $\lambda$ in $O(\nnz(\M))$ time and then preprocessing  $\M\leftarrow\M/\tilde{\lambda}$.

\begin{lemma}
\label{lem:square_root_rational}
Given any $\mu\I\preceq\M\preceq\I$ and $\eps\in(0,1)$, there is a rational function $r(x)$ with degree $k=O(\log(1/\eps)\log(1/\mu))$ satisfying %
$$
\|(r(\M)-\M^{1/2})\vv\|\le \eps\|\M^{1/2}\vv\|,\forall \vv\in\R^n.
$$	
\end{lemma}

In short, such a rational function can be denoted as $\hat{r}^{\mu}_k(x)$ and expressed as:
\begin{align}
    \hat{r}^\mu_k(x) = C x\prod_{i\in[k]}\frac{x+c_{2i}}{x+c_{2i-1}} 
    \text{ with }
c_i \defeq \mu\frac{\mathrm{sn}^2(\frac{i K'}{2k+1};\sqrt{\mu'})}{\mathrm{cn}^2(\frac{i K'}{2k+1};\sqrt{\mu'})}, i\in[2k].
\end{align}
\begin{align*}
\text{coefficients} & 
\begin{cases}
	C & \defeq \frac{2}{(\sqrt{\mu}\prod_{i\in[k]}\frac{\mu+c_{2i}}{\mu+c_{2i-1}})+(\prod_{i\in[k]}\frac{1+c_{2i}}{1+c_{2i-1}})},\\
c_i & \defeq \mu\frac{\mathrm{sn}^2(\frac{i K'}{2k+1};\sqrt{\mu'})}{\mathrm{cn}^2(\frac{i K'}{2k+1};\sqrt{\mu'})}, \forall i \in \{1, 2, \cdots,2k \}.
\end{cases}\\
\text{numerical constants} & 
\begin{cases}
	\sqrt{\mu'} & \defeq \sqrt{1-\mu},\\
     K' & \defeq \int_0^{\pi/2}\frac{d\theta}{\sqrt{1-\mu'\mathrm{sin}^2\theta}},\\
     u & \defeq F(\phi;\sqrt{\mu'})\defeq\int_0^\phi\frac{d\theta}{\sqrt{1-\mu'\sin^2\theta}},\\ 
     \mathrm{sn}(u;\sqrt{\mu'}) & \defeq\sin(F^{-1}(u;\sqrt{\mu'}));\ \mathrm{cn}(u;\sqrt{\mu'})\defeq\cos(F^{-1}(u;\sqrt{\mu'})).
\end{cases}
\end{align*}

It is easy to check that $\hat{r}^{\mu}_k(x)$ can be related to the rational function defined in~\cref{sec:PCP} through formula $\hat{r}^{\mu}_k(x^2)=x\cdot r^{\sqrt{\mu}}_k(x)$, where $r_k^\gamma(x)$ is defined as in~\eqref{eqn:Zolo} By \cref{cor:approx}, with the same order of $k\ge\Omega(\log(1/\eps)\log(1/\mu))$. A formal proof also utilizes such relationship between the rational functions.

\begin{proof}
Consider a modified rational function of Zolotarev rational as follows:
\begin{align}
    \hat{r}^\mu_k(x) = C x\prod_{i=1}^k\frac{x+c_{2i}}{x+c_{2i-1}} \label{eqn:Zolo-modified}
    \text{ with }
c_i \defeq \mu\frac{\mathrm{sn}^2(\frac{i K'}{2k+1};\sqrt{\mu'})}{\mathrm{cn}^2(\frac{i K'}{2k+1};\sqrt{\mu'})}, i\in[2k],
\end{align}
with coefficients $C$ as the corresponding $C$ in $r^{\sqrt{\mu}}_k(x)$, $\sqrt{\mu'}\defeq\sqrt{1-\mu}$ and $K' \defeq \int_0^{\pi/2}\frac{d\theta}{\sqrt{1-\mu'\mathrm{sin}^2\theta}}$, $\mathrm{sn}(u;\sqrt{\mu'})\defeq\sin(F^{-1}(u;\sqrt{\mu'}))$, $\mathrm{cn}(u;\sqrt{\mu'})\defeq\cos(F^{-1}(u;\sqrt{\mu'}))$ with definition $u=F(\phi;\sqrt{\mu'})\defeq\int_0^\phi\frac{d\theta}{\sqrt{1-\mu'\sin^2\theta}}$.

Note this rational actually satisfies the condition that $\hat{r}^{\mu}_k(x^2)=x\cdot r^{\sqrt{\mu}}_k(x)$. By \cref{cor:approx}, it holds that when $k\ge\Omega(\log(1/\eps)\log(1/\mu))$
\begin{equation}
\label{eqn:scalar}
|x\cdot r^{\sqrt{\mu}}_k(x)-x\cdot\sgn(x)|\le \eps|x|,\forall |x|\in[\sqrt{\mu},1].
\end{equation}

Now if we write $\M=\V\LLambda\V^\top$ where $\LLambda=\diag(\lambda_1,\cdots,\lambda_n)$ satisfying $1\ge\lambda_1\ge\cdots,\lambda_n\ge\mu$ and that each column of $\V$ is $\nnu_i,i\in[n]$. We can write $\vv=\sum\limits_{i=1}^n{\alpha_i\nnu_i}$, and thus get $\M^{1/2}\vv=\sum\limits_{i=1}^n \alpha_i\sqrt{\lambda_i}\nnu_i$.

Now we consider what we can get from substituting $x$ in $\hat{r}^\mu_k(x)$ with $\M$. By applying \cref{eqn:scalar} and substituting $x^2$ with $\M$, we get 
$$\|\hat{r}^\mu_k(\M)\vv-\sqrt{\M}\vv\|\le \eps\|\M^{1/2}\vv\|, \forall \vv\in\R^n.$$
\end{proof}

Now we give the following pseudocode in \cref{alg:SR} for $\SR(\M,\vv,\eps,\delta)$, which with probability $1-\delta$ outputs a solution $\xx$ satisfying $\|\xx-\M^{1/2}\vv\|\le\eps\|\M^{1/2}\vv\|$. In the pseudocode, we use $\LS(\M,c,\xx,\eps,\delta)$ to denote any solver that with probability $1-\delta$ gives as $\eps$-approximate solution $\tilde{\xx}$ of $\xx^*=(\M+cI)^{-1}\vv$ satisfying $\|\tilde{\xx}-\xx^*\|\le\eps\|\vv\|$.

\begin{algorithm}[H]
\DontPrintSemicolon
  \KwInput{$\M\in\R^{n\times n}$ data matrix, $\vv\in\R^n$ vector, $\eps$ accuracy, $\delta$ probability.}
  \Parameter{Smallest eigenvalue $\mu$, largest eigenvalue $\lambda$, degree $k$ (\cref{lem:square_root_rational}), coefficients $\{c_{i}\}_{i=1}^{2k},C$ (\cref{eqn:Zolo-modified}), accuracy $\eps_1$ (specified below)}
  \KwOutput{A vector $\xx$ satisfying $\|\xx-\M^{1/2}\vv\|\le\eps\|\M^{1/2}\vv\|$.}
	$\xx\gets\vv$\;
  \For{$i\gets 1$ to $k$}
  {
  $\xx\gets \M\xx+c_{2i}\xx$\;
  $\xx\gets \LS(\M,c_{2i-1},\xx,\eps_1,\delta/k)$\;
  }
  $\xx\gets C\cdot\M\xx$.\;
\caption{$\SR(\M,\vv,\eps,\delta)$}
\label{alg:SR}
\end{algorithm}

\begin{proof}[Proof of \cref{thm:square_root_solver_main}]
\ \\

Without loss of generality, we can assume $\lambda=1$, otherwise one can consider $\M/\lambda$ instead.

\paragraph{Choice of parameters:}
	
	We choose the following values for parameters in \cref{alg:SR}:
	\begin{align*}
	k & =\Omega(\log(1/\eps)\log(1/\mu))\\
	M & =\beta_3k^4/\beta_2\mu\\
	\eps_1 & = \frac{\eps}{8\beta_3k^3M^{k-1}}.
	\end{align*}

	Other coefficients $\{c_{i}\}_{i=1}^{2k},C$ are as defined in \cref{eqn:Zolo}. Here we use constants $\beta_2,\beta_3$ as stated in \cref{lem:bound}.
	
\paragraph{Approximation:}
	
	Given $\eps>0,\mu>0$, from \cref{lem:square_root_rational} we set $k\ge\Omega(\log(1/\eps)\log(1/\mu))$ thus $\hat{r}_k^{\mu}(x)$ as defined in \cref{eqn:Zolo-modified} satisfies $\|\M^{1/2}\vv-r_k^{\mu}(\M)\vv\|\le\eps/2\|\M^{1/2}\vv\|$. Now it suffices to get a $\xx$ satisfying $\|\xx-r_k^{\mu}(\M)\vv\|\le\eps/2\|\M^{1/2}\vv\|$.
		
Suppose we have a procedure $\mathcal{C}_i(\vv),i\in[k]$ that can apply $\CC_i\vv$ for arbitrary $\vv$ to $\eps'$-multiplicative accuracy with probability $\ge1-\delta/k$, where here $$\CC_i=\frac{\M+c_{2i}\I}{\M+c_{2i-1}\I}.$$ Also we assume matrix vector product is accurate without loss of generality.\footnote{If in the finite-precision world, we assume arithmetic operations are carried out with $\Omega(\log(n/\eps))$ bits of precision, the result is still true by standard argument with a slightly different constant factor for the bounding coefficient.} Note that 
\begin{align*}
	\left\|\frac{\M+c_{2i}\I}{\M+c_{2i-1}\I}\right\| & \le M,\forall i\in[k],
\end{align*}
with $M=\beta_3k^4/\beta_2\mu$. Here we use constants $\beta_2,\beta_3$ as stated in \cref{lem:bound}. Now we can use \cref{lem:stable} with the corresponding $M$ to show that: Using a union bound, with probability $\ge1-\delta$ it holds that
$$
	\|\mathcal{C}_{k}(\mathcal{C}_{k-1}(\cdots\mathcal{C}_1(\vv)))-\prod\limits_{i=1}^{k}\CC_i\vv\|\le 2\eps' kM^{k-1}\|\vv\|,
$$
whenever $\eps'\le\frac{M}{2k}$.

Now we choose $$\tilde{\eps}_1=\min(\frac{M}{2k},\frac{\sqrt{\mu}\eps}{8kM^{k-1}}),\ \eps_1=\min(\frac{M}{2k},\frac{\sqrt{\mu}\eps}{8kM^{k-1}})/(\beta_3k^2)=\frac{\eps\sqrt{\mu}}{8\beta_3k^3M^{k-1}},$$ consider the following procedure as in \cref{alg:ISPCP},
\begin{align*}
\xx & \gets \LS\left(\M,c_{2i-1},\xx,\eps_1,\delta/k\right); \forall i\in[k].\\
\xx & \gets C\cdot\M\xx.
\end{align*}

The above choice of $\eps_1$ guarantees procedures $\LS\left(\M,c_{2i-1},\xx,\eps_1,\delta/k\right)$ for all $i\in[k]$ can be abstracted as $\mathcal{C}_i(\vv)$ with $\tilde{\eps}_1$-accuracy and corresponding success probability. Using a union bound of successful events and the fact that $\|C\cdot \M\|\le 2$, we can argue that with probability $1-\delta$, the output $\xx$ of $\SR(\M,\vv,\eps,\delta)$ satisfy
$$
	\|\xx-r^\mu_k(\M)\xx\|\le 4\tilde{\eps}_1 kM^{k-1}\|\vv\|\le \frac{4\tilde{\eps}_1 kM^{k-1}}{\sqrt{\mu}}\|\M^{1/2}\vv\|\le\eps/2\|\M^{1/2}\vv\|,
$$
By \cref{dfn:square-root}, using the above choice of parameters we conclude this justifies the correctness argument in the theorem.

\paragraph{Runtime:} 
	
	Given the The numerical constants $C, \{c_{i}\}_{i=1}^{2k}$ are precomputed. So the runtime will then be a total runtime of computing matrix vector products for $k+1$ times, calling $k=O(\log(1/\eps)\log(1/\mu))$ times $\LS\left(\M,c_{2i-1},\xx,\eps_1,\delta/k\right)$ for $i\in[k]$. We bound the two terms respectively.
		
	Computing matrix vector products takes time $\tilde{O}(\nnz(\M))$ since $k=\tilde{O}(1)$. 
	
	Running time of  $\LS\left(\M,c_{2i-1},\xx,\eps_1,\delta/k\right)$ depends on the particular solver we use. Based on the assumption and the fact that $c_{2i-1}\in[\tilde{\Omega}(\mu),\tilde{O}(1)],\forall i\in[k]$, we can upper bound each solve by $\tilde{O}\left(\mathcal{T}\right)$.
		
	Adding it together gives running time of \cref{alg:SR} in
	$$
	\tilde{O}(\nnz(\M)+\mathcal{T}).
	$$
	
Replacing $\M$ with $\M/\lambda$, $\mu$ above with $\mu/\lambda$ and $\lambda$ with $1$ due to preprocessing gives the final statement.
\end{proof}

\section{More on Experiments}
\label{App:exp}

In this section we give a more detailed description and theoretical justification for $\emph{rlanczos}$ and $\emph{slanczos}$. Also we show its relationship and difference with our proposed algorithm $\ISPCP$.

\subsection{Details for \emph{rlanczos}}
\label{App:rlanczos}

Lanczos method is well known for being efficient and stable~\cite{musco2018stability}. As verified in theory and practice, running Lanczos algorithm on $(\A^\top\A+\lambda\I)^{-1}(\A^\top\A-\lambda\I)$ almost always beats the optimal universal approximation of $\sgn(x)$ by stably applying polynomial / chebyshev. This is because Lanczos can search for the best polynomial to approximate $\sgn(x)$ based on the distribution of $\A^\top\A$'s eigenvalues. 

Based on that, we also combined the well-studied rational Lanczos algorithm~\cite{gallivan1996rational} with the known Zolotarev rational expression, to search in the rational function space $r_{2k+1}(x)\in \mathcal{P}_{2k+1}/q_{2k}(x)$ where $q_{2k}(x)$ has exactly expression as in the denominator of $r^\gamma_k(x)$. Theoretically, this should always beat directly applying Zolotarev rational by allowing more freedom to cater to the distribution of eigenvalues of $\A^\top\A$.

The two methods \emph{rational} and \emph{rlanczos} are quite close in practice. (see \cref{fig:gap,fig:n}) In general for low accuracy regime, \emph{rlanczos} slightly improve on \emph{rational} $\ISPCP$. But $\ISPCP$ is more stable and can get to more accurate solutions. This also shows the strength of $\ISPCP$ proposed in the paper.

\subsection{Details for \emph{slanczos}}
\label{App:slanczos}

For \emph{slanczos}, the idea is to incorporate the squared system primitive with lanczos method on polynomial directly, i.e. searching for function in form $$f\left(\frac{(x-\lambda)(x+\lambda)}{(x-\lambda)^2+\gamma(x+\lambda)^2}\right)$$ and replace $x\leftarrow \A^\top\A$.

Note here we introduce shift-and-rescaling $(\A^\top\A+\lambda\I)^{-1}(\A^\top\A-\lambda\I)$ so that all eigenvalues of $\A^\top\A$ satisfying $\lambda_i\in[0,\lambda(1-\gamma)\cup(\lambda(1+\gamma),\infty)$ are mapped to range $[-1,-\gamma/2]\cup[\gamma/2,1]$. So for now let's consider $|x|\in[\gamma/2,1]$. One observation is that there is this squared primitive in form $x/(x^2+\gamma)$ can map $|x|\in[\gamma/2,1]$ to $[\Theta(1),\Theta(1/\sqrt{\gamma})]$.

\begin{lemma}\label{lem:rational}
Take $r(x)=x/(x^2+\gamma)$. Then $\tilde{r}(x)=2\sqrt{\gamma}r(x)$ maps $x\in[\gamma/2,1]$ to $(2\sqrt{\gamma}/3,1)$ and $x\in[-1,-\gamma/2]$ to $(-1,-2\sqrt{\gamma}/3)$.
\end{lemma}
\begin{proof}

We only consider $x\in[\gamma/2,1]$ and a exactly symmetric argument works for the other side. Now consider $1/r(x)=x+\gamma/x$, we have when $x\in[\gamma/2,1]$, $1/r(x)\in[2\sqrt{\gamma},2+\gamma]$, thus showing $r(x)\in[1/(2+\gamma),1/(2\sqrt{\gamma})]\subseteq(1/3,1/2\sqrt{\gamma}).$ Multiplying by coefficient we have $\tilde{r}(x)\in(2\sqrt{\gamma}/3,1)$ whenever $x\in[\gamma/2,1]$, which completes the proof.
\end{proof}

\begin{remark}
For such a rational primitive $\tilde{r}(x)$ since we know there is optimal $O(\mathrm{log}(1/\eps)/\sqrt{\gamma})$-degree chebyshev polynomial $f(x)$ that maps $[\sqrt{\gamma},1]$ to $[1-\eps,1]$ and $[-1,-\sqrt{\gamma}]$ to $[-1,-1+\eps]$, thus we are able to run only $\tilde{O}(1/\sqrt{\gamma})$ suboracles of ridge regression of solving $((\A^\top\A-\lambda \I)^2+\gamma (\A^\top\A+\lambda \I))\xx=\vv$.\end{remark}

Formally combining this with the squared system solver we develop in \cref{thm:square_solver_main} leads to the following theoretical guarantee:

\begin{theorem}[Runtime for \emph{slanczos}]
\label{thm:slanczos}
\emph{slanczos} can be converted into an $\eps$-approximate PCP / PCR solver with runtime guarantee
 $$\tilde{O}\left(\nnz(\A)/\sqrt{\gamma}+\sqrt{\nnz(\A)d\cdot\sr(\A)}\kappa/\gamma\right).$$	
\end{theorem}

This guarantee implies the method \emph{slanczos} would work better than both \emph{polynomial}, \emph{chebyshev} in~\cite{FMMS16,AL17}, and rational methods including $\ISPCP$ and \emph{rlanczos} (see \cref{App:rlanczos}) for certain regimes of parameters.

As the runtime of \emph{slanczos} is not almost linear, we don't state it formally in the main part of paper. Also, as $\nnz(\A)/d\cdot\sr(\A)^2\kappa^2$ or simply $n/d$ gets larger, it gets worse compared with the two purely rational methods. This also is verified by our experiments shown in \cref{fig:n} - when we fix $d,\gamma$ and increase the magnitude of $n$, \emph{rational} and \emph{rlanczos} start outperforming \emph{slanczos}.

\end{document}